\documentclass[letterpaper]{article}
\usepackage{aaai25}
\usepackage{times}
\usepackage{helvet}
\usepackage{courier}
\usepackage[hyphens]{url}
\usepackage{graphicx}
\usepackage{natbib}
\usepackage{caption}
\usepackage{algorithm}
\usepackage{algorithmic}

\usepackage{newfloat}
\usepackage{listings}
\DeclareCaptionStyle{ruled}{labelfont=normalfont,labelsep=colon,strut=off}
\floatstyle{ruled}
\newfloat{listing}{tb}{lst}{}
\floatname{listing}{Listing}
\title{Distances Between Top-Truncated Elections of Different Sizes}
\author {
    Piotr Faliszewski\textsuperscript{\rm 1},
    Jitka Mertlová\textsuperscript{\rm 2},
    Pierre Nunn\textsuperscript{\rm 3},
    Stanisław Szufa\textsuperscript{\rm 4},
    Tomasz Wąs\textsuperscript{\rm 5}
}
\affiliations {
    \textsuperscript{\rm 1}AGH University\\
    \textsuperscript{\rm 2}Czech Technical University in Prague\\
    \textsuperscript{\rm 3}Université de Rennes\\
    \textsuperscript{\rm 4}CNRS, LAMSADE, Université Paris Dauphine-PSL\\
    \textsuperscript{\rm 5}University of Oxford\\
    faliszew@agh.edu.pl,
    mertlova.jita@gmail.com, 
    pierre.nunn@ens-rennes.fr,
    s.szufa@gmail.com,
    tomasz.was@cs.ox.ac.uk
}

\usepackage{times}
\usepackage{soul}
\usepackage{url}

\usepackage[utf8]{inputenc}

\usepackage{graphicx}
\usepackage{amsmath}
\usepackage{amsthm}
\usepackage{booktabs}
\usepackage{algorithm}
\usepackage{algorithmic}

\usepackage{t1enc}

\usepackage{booktabs}

\usepackage{graphicx} 
\usepackage{fontawesome5}
\usepackage{subcaption}

\usepackage{amssymb}
\usepackage{amsmath}
\usepackage{mathtools}

\usepackage{cleveref}

\newtheorem{theorem}{Theorem}[section]
\newtheorem{lemma}[theorem]{Lemma}
\newtheorem{proposition}[theorem]{Proposition}
\newtheorem{corollary}[theorem]{Corollary}

\usepackage{xspace}
\usepackage{nicefrac}

\usepackage{pifont}
\newcommand{\ymark}{\textcolor{mygreen}{\ding{51}}}%
\newcommand{\nmark}{\textcolor{myred}{\ding{55}}}%

\usepackage{tikz}
\usetikzlibrary{decorations.pathreplacing}
\usetikzlibrary{arrows.meta}
\tikzset{>={Latex[width=1mm,length=3mm]}}

\newcommand{\swap}{\mathrm{swap}}

\newcommand{\del}{\mathrm{del}}
\newcommand{\trun}{\mathrm{\,tr}}
\newcommand{\pos}{\mathrm{pos}}
\newcommand{\dap}{\mathrm{dap}}
\newcommand{\stretched}{\mathrm{str}}
\newcommand{\emd}{\mathrm{emd}}
\newcommand{\freq}{\mathrm{freq}}
\newcommand{\ID}{\mathrm{ID}\xspace}
\newcommand{\AN}{\mathrm{AN}\xspace}
\newcommand{\UN}{\mathrm{UN}\xspace}
\newcommand{\ST}{\mathrm{ST}\xspace}
\newcommand{\id}{\mathrm{id}\xspace}
\newcommand{\an}{\mathrm{an}\xspace}
\newcommand{\un}{\mathrm{un}\xspace}

\newcommand{\reals}{\mathbb{R}}

\newcommand{\posreals}{\mathbb{R}_{+}}
\newcommand{\normphi}{{{\mathrm{norm}\hbox{-}\phi}}}

\newcommand{\pref}{\succ}

\newcommand{\tabimg}[1]{\raisebox{-0.1cm}{\includegraphics[width=0.4cm]{#1.png}}}

\usepackage{rotating}

\newcommand{\citetapp}[1]{\citet{#1}}
\newcommand{\citepapp}[1]{\citep{#1}}

\newcommand{\calT}{\mathcal{T}}

\usepackage[textsize=scriptsize]{todonotes}

\usepackage{xcolor}

\usepackage{etoolbox}

\newcommand{\appendixProofs}{}

\newcommand{\appendixproof}[3]{%
  \gappto{\appendixProofs}{
		\subsection{Proof of~{#1}~\ref{#2}}\label{proof:#2}
    #3
  }
}

\begin{document}
\definecolor{myred}{RGB}{201, 22, 22}
\definecolor{mygreen}{RGB}{38, 150, 68}

\maketitle

\begin{abstract}
  The map of elections framework is a methodology for visualizing and
  analyzing election datasets. So far, the framework was restricted to
  elections that have equal numbers of candidates, equal numbers of
  voters, and where all the (ordinal) votes rank all the
  candidates. We extend it to the case of elections of different
  sizes, where the votes can be top-truncated. We use our results to
  present a visualization of a large fragment of the Preflib database.
\end{abstract}

\begin{links}
    \link{Code}{https://github.com/Project-PRAGMA/Map-Different-Sizes-AAAI-2025}
\end{links}

\section{Introduction}

The map of elections framework is a methodology for analyzing and
visualizing sets of elections,
introduced by
\citet{szu-fal-sko-sli-tal:c:map} and
\citet{boe-bre-fal-nie-szu:c:compass}.  The basic idea is to depict
(ordinal) elections
as points on a plane, so that the closer two given points are, the
more similar are their corresponding elections (see
Figures~\ref{fig:map:poswise} and \ref{fig:map:swap} for examples, and \Cref{sec:prelim} for definitions). To measure this similarity, one
can either use the
accurate but computationally challenging isomorphic swap
distance~\citep{fal-sko-sli-szu-tal:c:isomorphism}, or the
less precise but efficiently computable positionwise
distance~\citep{szu-fal-sko-sli-tal:c:map}, or whatever other distance
that is invariant to renaming the candidates and voters.
Indeed,
the elections may be
unrelated to each other and we only care about their structural
similarities.
Unfortunately, the distances used so far
are restricted to elections with the same numbers of candidates and
voters, and require the votes to be complete,
i.e., to rank all the candidates.
We aim to rectify these two issues.

The original motivation behind the map framework was to better
understand relations between various statistical cultures (i.e.,
models of generating random elections) and to present experimental
results in a nonaggregate way~\citep{szu-fal-sko-sli-tal:c:map}.
For this,
the restriction to particular election sizes
and complete votes is natural, as we have full control over the
data.
However,
maps are also useful for studying real-life elections. For example, we
can get better insight into the nature of real-life elections by
analyzing their positions on the
map~\citep{boe-bre-fal-nie-szu:c:compass,boe-sch:c:real-life-elections},
or check which statistical cultures yield similar
elections~\citep{boe-bre-fal-nie-szu:c:compass,boe-bre-elk-fal-szu:c:frequency-matrices}.
Yet,
real-life data poses some issues, as seen in the next two
datasets from Preflib~\citep{mat-wal:c:preflib}:

\begin{description}
\item[Irish General Elections.] Preflib includes the ordinal votes
  cast in several constituencies in Ireland during the 2002 general
  election. Naturally, each of them
  had a different
  number of candidates (between 9 and 14) and a different number of
  voters (between around $30\ 000$ and $64\ 000$),
  most of whom
  ranked some top candidates only.
\item[Formula~1.]
  Here, each election is
  a Formula~1 season, where each
  vote is a race ranking the drivers (i.e., the candidates) in the
  order in which they finished. There are around 20 races and 20
  drivers in each season, but these numbers vary,
  and not all drivers complete each race.
\end{description}
So far, to put such elections on a map, one had to fill in the missing
parts of the preference orders, restrict the number of candidates to
some common value (e.g., by deleting the worst-performing ones), and
sample a fixed-size set of votes.  We propose distances that can deal
with different-sized/top-truncated elections natively, without
preprocessing.

\begin{figure*}[t]
  \centering
  \hfill
  \begin{subfigure}[b]{0.3\linewidth}
    \includegraphics[width=\textwidth]{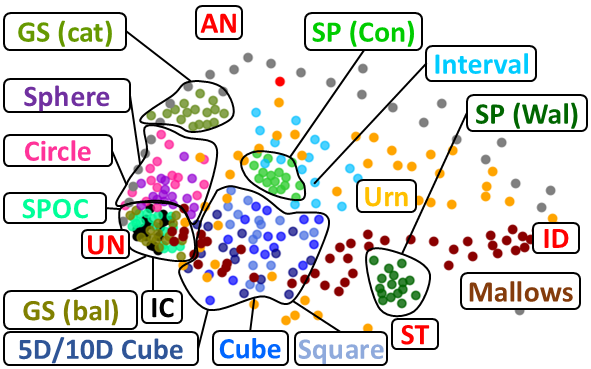}%
    \caption{Positionwise Distance.}
    \label{fig:map:poswise}
  \end{subfigure}
  \hfill
  \begin{subfigure}[b]{0.3\linewidth}
    \includegraphics[width=0.85\textwidth]{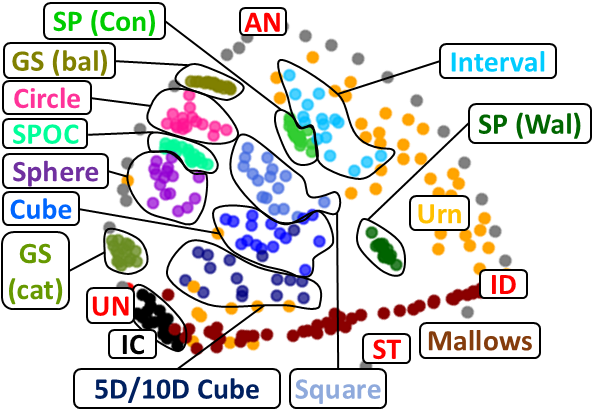}%
    \caption{Swap Distance.}
    \label{fig:map:swap}
  \end{subfigure}
  \hfill
  \begin{subfigure}{0.01\linewidth}
  \end{subfigure}
  \caption{\label{fig:basic-maps}Maps of elections created using the
    (a)~isomorphic swap and (b)~positionwise distances. Each point
    corresponds to an election; the color represents the statistical
    culture it comes from: $\ID$, $\UN$, and $\AN$ refer to identity,
    uniformity, and antagonism elections; IC, Mallows, and urn mean
    impartial culture, the normalized Mallows model, and the
    Pólya-Eggenberger urn model, respectively; Interval, Square, Cube,
    5D/10D Cube, Circle, and Sphere refer to Euclidean models;
    SP stands for single-peaked (in the Conitzer or Walsh
    models), SPOC for single-peaked on a circle, and GS for
    group-separable (caterpillar or balanced).
    }
\end{figure*}

Ideally, we would like to extend the isomorphic swap and positionwise
distances so that different-sized elections with ``obviously
identical'' structure would be at distance zero. Unfortunately, this
turns out to be impossible even for a rather minimal definition of an
``obviously identical'' structure (even without thinking of
top-truncated votes). However, we do find two intuitively appealing
distances that 
fulfill some of our requirements: One is an extension of the
positionwise distance and one, which we call DAP, is based on
analyzing diversity, agreement, and polarization within votes. In
the spirit of our desiderata, DAP turns out to be strongly correlated with
the isomorphic swap distance and, so, can serve as its replacement.
We first describe our distances and then
analyze their properties, both theoretically and experimentally.
In the experiments, we largely focus on synthetic
elections,
as such data is easier to control, but
our main motivation is to visualize real-life data. Hence, in
\Cref{sec:map-of-preflib} we form a map of a large fragment of
Preflib, a database of real-life elections (see
\Cref{fig:map:preflib}). The overarching goal of the paper is to
explain how we obtained this map, why the choices made on the way are
justified, and what we can learn from it.

\section{Preliminaries}\label{sec:prelim}

For a positive integer $k$, we let $[k] = \{1,\ldots, k\}$, and we
write $\posreals$ to denote the set of nonnegative real numbers.  For
two equal-sized sets $A$ and $B$, $\Pi(A,B)$ is the set of bijections
between $A$ and $B$.

\paragraph{Elections.}
An \emph{election} is a pair $E = (C,V)$, where $C$ is a set of $m$
\emph{candidates} and $V$ is a collection of $n$ \emph{votes}, i.e.,
weak orders over $C$; we will also sometimes refer to
members of $V$ as \emph{voters} (i.e., the agents who cast the
respective votes).  By a \emph{size} of an election, we mean the
pair~$(m,n)$.
For two candidates $a,b \in C$ and voter $v \in V$, we write
$a \succ_v b$ if $v$ strictly prefers $a$ over $b$ and $a \sim_v b$ if
$v$ is indifferent between them.  We write $a \succeq_v b$ if
$a \succ_v b$ or $a \sim_v b$. By $N_E(a \succ b)$ we mean the number
of voters in election $E$ that strictly prefer $a$ to $b$, and by
$N_E(a \sim b)$ we mean the number of those indifferent between these
candidates.

Consider
vote $v$ over candidate set $C$. We say that $v$
is \emph{complete} if it is a strict linear order over $C$, and that
it is \emph{top-truncated} if $C$ can be partitioned into 
sets
$C_v^{\uparrow}$ and $C_v^{\downarrow}$, so that $v$ is a strict
linear order over $C_v^{\uparrow}$ and for every
$c \in C_v^{\uparrow}$ and $a, b \in C_v^{\downarrow}$ we have
$c \succ_v a \sim_v b$.  $C_v^{\uparrow}$ is the top part
of $v$ and $C_v^{\downarrow}$ is its truncated part (so a complete vote
is top-truncated with empty truncated
part.)  An election is complete (top-truncated) if all its
votes are complete (top-truncated).

Top-truncated elections
are quite common in real-life election data
(see e.g., \emph{Preflib} database~\cite{mat-wal:c:preflib})
and have been extensively studied
in terms of their prevalence \cite{kilgour2020prevalence},
efficiency \cite{borodin2022distortion},
effect on voting rules \cite{tomlinson2023ballot},
or a possibility of manipulation
\cite{bau-fal-lan-rot:c:lazy-voters}.

\paragraph{Special Elections.} Following
\citet{boe-bre-fal-nie-szu:c:compass}, we consider the following
families of characteristic (complete) elections (we consider $m$
candidates and $n$ voters):
  (i) In an \emph{identity election} ($\ID_{m,n}$) all $n$ votes are
  identical;
  (ii) an \emph{antagonism election} ($\AN_{m,n}$) is a concatenation
  of two identity ones, where the voters in one have opposite
  preference orders to those in the other; and
  (iii) in a \emph{uniformity election} ($\UN_{m,n}$) all possible votes
  appear equal number of times.
The number of voters must be even in an antagonism election
and must be a multiple of $m!$ in a uniformity one.
We
drop the number of voters from notation when it is
irrelevant.

\paragraph{Distances.}
A \emph{pseudodistance} over a set $U$ is a function
$d \colon U \times U \rightarrow \posreals$ such that for each
$a,b, c \in U$ we have (1)~$d(a,a) = 0$, (2)~$d(a,b) = d(b,a)$, and
(3)~$d(a,b)+d(b,c) \geq d(a,c)$. We often drop ``pseudo''  prefix.

\paragraph{(Isomorphic) Swap Distance.}
Given two votes, $u, v$, over the same candidate set, we define
their swap
distance as
$\swap(u,v) = \frac{1}{2}\sum_{a,b \in C} ([a \succ_u b \ \land \ b
\succeq_v a] + [a \succeq_u b \ \land \ b \succ_v a])$, where the
expressions in square brackets evaluate to $1$ when they are true and
to~$0$ otherwise.
That is, we add $1$ for each pair of candidates for which both voters
have a strict, opposite preference, and we add $\nicefrac{1}{2}$ for
each pair of candidates for which exactly one of the voters is
indifferent.

For a pair of equal-sized/complete elections $E=(C,V)$ and $F=(D,U)$,
such that $|C|=|D|=m$, $V = \{v_1, \dots, v_n\}$ and
$U = \{u_1, \dots, u_n\}$, their \emph{(isomorphic) swap distance} is
the sum of the swap distances between the votes, under optimal matchings
of the candidates and voters.  Formally, $d_{\swap}(E,F)$ is equal to:
\[
     \min_{\substack{\pi \in \Pi([n],[n])\\ \sigma \in \Pi(C,D)}}
        \textstyle\sum_{i \in [n]} \swap(\sigma(v_i), u_{\pi(i)})/\left( \textstyle \frac{1}{4}n(m^2-m) \right),
\]
where $\sigma(v_i)$ means vote $v_i$ with each candidate $c \in C$
replaced with $\sigma(c) \in D$.  By definition, $d_\swap$ is
invariant to renaming the candidates and reordering the voters, and
its value is zero
if and only if the given elections are isomorphic (i.e., are identical
up to renaming the candidates and reordering the voters).  This
distance was introduced by \citet{fal-sko-sli-szu-tal:c:isomorphism}.
Dividing by $\frac{1}{4}n(m^2-m)$ ensures normalization of the largest
distance to exactly $1$ (for the case where $n$ is
even~\citep{boe-fal-nie-szu-was:c:metrics}).

\paragraph{Wasserstein Distance.}

 Given a vector $\vec{a}\! =\! (a_1, \ldots, a_m) \in \posreals^m$, we
 identify it with a stepwise function
$a \colon [0,1] \rightarrow \posreals$, such that $a(0) = 0$ and for
all $i \in [m]$ and $x$ in $(\frac{i-1}{m},\frac{i}{m}]$ we have
$ a(x) = (m\cdot a_i)/(a_1 + \cdots + a_m)$.  Further, for each
$x \in [0,1]$ we set $A(x) = \int_0^x a(y)dy$. Note that the
normalization in function $a(\cdot)$ ensures that $A(1) = 1$.
For two vectors $\vec{a} \in \reals^s$ and $\vec{b} \in \reals^{t}$ of
possibly different dimensions, we define their \emph{Wasserstein
  distance} as
$
  \textstyle
  W(\vec{a},\vec{b}) = 
  \int_{0}^{1} |A(x) - B(x)| dx.
$
As $a(x)$ and $b(x)$ are stepwise 
and $A(x)$ and $B(x)$ are piecewise linear, $W(\vec{a},\vec{b})$ can
be computed in polynomial time (assuming we can perform arithmetic
operations in polynomial time).

\paragraph{Frequency Matrices and Positionwise Distance.}
Consider an election $E = (C,V)$, where $C = \{c_1, \ldots, c_m\}$ and
$V = (v_1,\ldots,$ $ v_n)$.  A \emph{frequency matrix} of a
top-truncated vote $v \in V$ is an $m \times m$ 
matrix $\freq(v)$ where the entry in the $i$-th row and $j$-th column
is:
\[
 \freq(v)_{i,j} =
 \begin{cases}
    1 & \!\mbox{if } c_j \in C_v^{\uparrow}
        \mbox{ and $v$ ranks it $i$-th there},  \\ 
    1 / |C_v^{\downarrow}| & \!\mbox{if } c_j \in C_v^{\downarrow}
        \mbox{ and } i > |C^{\uparrow}|,\\
    0 & \!\mbox{otherwise}.
 \end{cases}
\]
A frequency matrix of election $E$ is the average of the frequency
matrices of its votes, i.e.,
$\freq(E) = \nicefrac{1}{n}\sum_{v_i \in V}\freq(v_i)$.  Intuitively,
$\freq(E)_{i,j}$ is the probability that in a randomly selected vote
from $E$ the $j$-th candidate is ranked in the $i$-th position (where
being in the truncated part of a vote means having equal chance of
being ranked in any of the truncated positions).  Frequency matrices
are normalized variants of \emph{position matrices} of
\citet{szu-fal-sko-sli-tal:c:map}; see also the works of
\citet{boe-bre-elk-fal-szu:c:frequency-matrices,boe-cai-fal-fan-jan-kac:c:position-matrices}
(these authors did not consider top-truncated votes; we added this
extension).  Frequency matrices are bistochastic, i.e., each of their
columns and each of their rows sums up to $1$.

The Wasserstein distance between two bistochastic matrices,
$X = [\vec{x}_1, \ldots, \vec{x}_m]$
and $Y = [\vec{y}_1, \ldots, \vec{y}_m]$, with $m$ column vectors
each, is the sum of the Wasserstein
distances between their (optimally matched) column vectors.  Formally:
\[
  \textstyle
  d_{W}(X,Y) = \min_{\sigma \in \Pi([m],[m])}
        \frac{1}{m}\sum_{i \in [m]}
        W(\vec{x}_i,\vec{y}_{\sigma(i)}).
\]
Then, the positionwise distance of two equal-sized elections, $E$ and $F$, is
the Wasserstein distance of their frequency matrices,
$d_\pos(E,F) = d_W(\freq(E),\freq(F))$.
This distance was originally defined by
\citet{szu-fal-sko-sli-tal:c:map} using the classic earth mover's
distance (EMD) of \citet{rub-tom-gui:j:emd}, defined for vectors of
the same dimension that sum up to the same value.
We prefer the Wasserstein distance because it applies to vectors of
different dimensions, while being closely related to EMD (in the
proposition below, $\emd(\vec{a},\vec{b})$ is the EMD distance between
vectors and $\vec{a}$ and $\vec{b}$, and is known to be between $0$
and $m-1$; hence $m W(\vec{a},\vec{b})$ is its good approximation
both for very similar and quite far-off vectors).
\begin{proposition}\label{prop:emd-vs-w-in-text}
  For each two $m$-dimensional vectors $\vec{a}$ and $\vec{b}$ with
  nonnegative entries that sum up to $1$, it holds that
  $\max(\frac{1}{2}\emd(\vec{a},\vec{b}),\emd(\vec{a},\vec{b})\! -\!
  1) \leq m\!\cdot\! W(\vec{a},\vec{b}) \leq
  \emd(\vec{a},\vec{b})$.
\end{proposition}

Like the isomorphic swap distance, the positionwise distance is
invariant to renaming the candidates and reordering the voters, but it
differs in that it can assume value $0$ even for pairs of
nonisomorphic elections.  Positionwise distance is normalized in the
sense that its largest value is about $\nicefrac{1}{3}$ (we omit exact
calculations).

\newcommand{\myfrac}[2]{\left(#1 \right)/\left(#2\right)}

\paragraph{Maps of Elections.} A map of elections is a set of
elections with distances between each pair of them (called
\emph{original distances}). We associate each election with a point in
a 2D space so that the Euclidean distances between the points resemble
the original ones, using either multidimensional scaling
(MDS~\citep{kru:j:mds}) or Kamada-Kawai
(KK~\citep{kam-kaw:j:embedding}) embedding algorithms.

\section{Search for More Versatile Distances}\label{sec:distances}

Our first goal is to find a satisfying distance $\hat{d}$ over
different-sized, complete elections (as a convention, we will be
marking distances that can handle different-sized elections with a
``hat'' on top).  
Ideally, we would like $\hat{d}$ to be an extension of
an appealing distance among equal-sized
elections,
and to put ``obviously'' identical elections at distance $0$.
Below we express these
requirements formally:

\begin{enumerate}
\item We say that $\hat{d}$ is a \emph{swap extension} if it is a
  distance and
  for each two
  equal-sized/complete elections $E_1$ and $E_2$ with $n$ voters and
  $m$ candidates, 
  $\hat{d}(E_1,E_2) = d_\swap(E_1,E_2)$.
  We define a \emph{positionwise extension} analogously.

\item We say that $\hat{d}$ is
  \emph{$\ID$-consistent} if for all $p,r,s,q \in \mathbb{N}$, $p,r
  \ge 3$, we have $\hat{d}(\ID_{p,q},\ID_{r,s}) =
  0$.
  We define
  \emph{$\AN$-consistency} and
  \emph{$\UN$-consistency} analogously.  (We consider at least
  $3$ candidates as $\AN_{2,2} =
  \UN_{2,2}$, so without this restriction assuming both
  $\AN$- and $\UN$-consistency would put all $\AN$ and
  $\UN$ elections at distance zero.)

\end{enumerate}

We would like a
swap extension
or a positionwise extension 
that is $\ID$-, $\AN$-, and $\UN$-consistent.  Naturally, one could
think of further conditions, but these are both natural and quite
minimal. Yet, even this is too much to ask for.

\begin{proposition}\label{prop:swap-pos:impossibility}
  There is no swap extension nor positionwise extension that
  simultaneously satisfies $\ID$\hbox{-}, $\AN$- and
  $\UN$-consistency.
\end{proposition}

This means that
either we have to give up on at least one of the consistency
properties, or 
on looking for swap/positionwise extensions.  We explore both of these
possibilities.

\subsection{Intuitive Swap/Positionwise Extensions}
\label{sec:swap-positionwise}

As far as swap extensions go, one natural idea
is that given two elections $E_1 = (C_1, V_1)$
and $E_2 = (C_2,V_2)$, where $|C_1| < |C_2|$, we define
$\hat{d}^\trun_\swap(E_1,E_2)$ as $d_\swap(E'_1,E_2)$, where
$E'_1$ is equal to $E_1$ with additional $|C_2|-|C_1|$ candidates that
all voters have in their truncated parts. Unfortunately,
this leads to maps where elections with fewer candidates are clustered
together, so we lose significant amount of information
about them.
\begin{proposition}
  $\hat{d}^\trun_\swap$ is a swap extension that is neither $\ID$- nor
  $\AN$- nor $\UN$-consistent.
\end{proposition}
Alternatively, we could define $\hat{d}^\del_\swap(E_1,E_2)$ to be
equal to the expected isomorphic swap distance of elections $E_1$ and
$E'_2$, where $E'_2$ is obtained by deleting $|C_2|-|C_1|$ candidates
from $E_2$ uniformly at random. However, this is not even a distance as it
fails the triangle inequality.  

\begin{proposition}\label{prop:swap-triangle}
  $\hat{d}^\del_\swap$ fails triangle inequality.
\end{proposition}

One could hope that violations of the
triangle inequality
are rare and, hence, could be ignored.
Unfortunately,
based on the experiments,
this is not the case,
and violations are common.
All in all, we were not able to find a
satisfying swap
extension and we leave looking for one as an open problem.
On the positive side, we do
identify a natural, $\UN$-consistent positionwise extension that not
only applies to different-sized elections, but also seamlessly handles
top-truncated votes.

For a matrix $X = [\vec{x}_1, \ldots, \vec{x}_m]$ and an integer
$k \geq 1$, by $\stretched_{mk}(X)$ we denote the matrix obtained from
$X$ by \emph{stretching}, i.e., copying each of its component vectors $k$ times, so that for
each $i \in [mk]$, the $i$-th vector of $\stretched_{mk}(X)$ is
$\vec{x}_{\lceil i/k \rceil}$.
Then, we define the positionwise distance between elections $E$ and $F$---with possibly
different numbers of candidates and voters---as
$\hat{d}_\pos(E,F) = d_W(\stretched_s(\freq(E)),\stretched_s(\freq(F))),$
where $s$ is the least common multiple 
of the numbers of candidates in $E$ and~$F$.
That is, to obtain $\hat{d}_\pos(E,F)$ we first compute the frequency
matrices of $E$ and $F$, then we duplicate their vectors so that both
matrices end up with an equal number of columns, and, finally, we
match these columns and sum up their Wasserstein distances (recall the
definition of the positionwise distance from \Cref{sec:prelim}; note
that here the vectors may be of different dimensions, which
is why we chose to use the Wasserstein distance instead of EMD).
Truncated votes are handled seamlessly because they are encoded in the
frequency matrices.  We show that $\hat{d}_\pos$ is indeed a
positionwise extension
and that it is UN-consistent.

\begin{theorem}\label{thm:pos-ext}
  $\hat{d}_\pos$ is a positionwise extension that is $\UN$-consistent,
  but not $\ID$- nor $\AN$-consistent.
\end{theorem}

Note that \Cref{prop:swap-pos:impossibility} does not
not preclude the existence of a positionwise extension satisfying $\UN$-consistency and either $\ID$- or $\AN$-consistency. 
Thus, our positionwise distance is not ``optimal'' in terms of the axioms it satisfies.

\subsection{Feature Distance and ID/AN/UN-Consistency}\label{sec:feature}

In principle, \Cref{prop:swap-pos:impossibility} might hold simply
because no distance is simultaneously $\ID$-, $\AN$-, and
$\UN$-consistent, irrespective of 
being a swap or a
positionwise extensions. We show that this is not the case and at
least one such distance exists.
To this end, let $f = (f_1,\ldots,f_k)$ be a collection of
\emph{features}, i.e., functions that given an election output a value
between $0$ and $1$. For an election $E$, its feature vector is
$f(E) = (f_1(E),\ldots,f_k(E))$, and the \emph{feature distance}
$\hat{d}_f$ between elections $E$ and $F$ is
$d_f(E,F) = \ell_2(f(E),f(F))$; naturally, one could also use $\ell_1$
or some other distances.%
\footnote{ For two vectors, $\vec{a}=(a_1,\dots,a_k)$ and
  $\vec{b}=(b_1,\dots,b_k)$, we define
  $\ell_2(\vec{a},\vec{b})=\left(\sum_{i=1}^k
    (a_i-b_i)^2\right)^{\nicefrac{1}{2}}$.}  
    For
every collection of features $f$, $\hat{d}_f$ is indeed a distance
over different-sized elections, and if the features are defined for
top-truncated elections, so is 
$\hat{d}_f$.

Let $\id(E)$, $\an(E)$, and $\un(E)$ be three 
features such that their value is~$0$ if the input election is
isomorphic to, respectively, some $\ID$, $\AN$, or $\UN$ election, and
it is $1$ otherwise.

\begin{proposition}
  The distance defined by collection $(\id, \an, \un)$ of features 
  is $\ID$-, $\AN$-, and $\UN$-consistent.
\end{proposition}
While this distance has all the desired consistency properties, it puts
all elections that are not isomorphic to $\ID$, $\AN$, or $\UN$ at
distance zero and, hence, is not very useful.
Next we develop a feature distance that
implements a similar idea, but in a more sophisticated way.
To this end, 
we focus on 
evaluating diversity, agreement, and polarization among the votes
(indeed, $\UN$ captures ideal diversity, $\ID$ ideal agreement, and
$\AN$---polarization).  Let $E = (C,V)$ be an election.  The agreement
for candidates $a,b \in C$ is:
\[
  \alpha(a,b) = \max( |N_E(a\!\succ\! b) - N_E(b\!\succ\! a)|, N_E(a \!\sim\! b)) /
  |V|.
\]
Intuitively, $|N_E(a\succ b) - N_E(b\succ a)|/|V|$ captures the
agreement when voters lean toward strict preference over $a$ and $b$,
and $N_E(a \sim b)/|V|$ 
reflects the agreement when most of the voters
put both $a$ and $b$ in the truncated parts of their votes.
The agreement
index of an election $E=(C,V)$ is an average of agreements for all
pairs of candidates, i.e.,
$A(E) = \sum_{\{a,b\} \subseteq C} \alpha(a,b) / \binom{|C|}{2}.$
This definition is a natural extension of the agreement index studied for complete elections by
\citet{alc-vor:j:cohesiveness}, \citet{has-end:c:diversity-indices}, and
\citet{can-ozk-sto:polarization}.
Note that these papers use other
names and interpretations for this index---%
we follow the approach of
\citet{fal-kac-sor-szu-was:c:microscope}.
\newcommand{\empkem}{{{\mathrm{emk}}}}

Regarding diversity and polarization, we also follow
\citet{fal-kac-sor-szu-was:c:microscope}, but with a few
changes. 
For each integer $i$, let the empirical $i$-Kemeny score of election
$E$ be:
\[
  \textstyle
  \empkem_i(E) = \min_{v_1,\ldots, v_i \in V} \! \left( \sum_{v \in V} \left( \min_{j \in [i]} \swap(v,v_j) \right) \right)\!.
\]
That is, to compute $\empkem_i(E)$, we seek $i$ votes from the
election so that the sum of the swap distances of each vote in $E$ to
the closest selected one is minimized. E.g., $\empkem_1(E)$
is an approximation of the classic Kemeny
score~\citep{kem:j:no-numbers}. To normalize $\empkem_i(E)$, we divide
by its maximal possible value, i.e.,
$\nicefrac{1}{2}\cdot |V| \cdot \binom{|C|}{2}$.  We define the
diversity and polarization indices as:
\begin{align*}
  \textstyle
  D(E) &= \textstyle\frac{2}{5}\sum_{i=1}^5\empkem_i(E)/(|V| \cdot \binom{|C|}{2}),\\
  \textstyle P(E) &=\textstyle 2 \cdot (\empkem_1(E) - \empkem_2(E))/(|V| \cdot
  \binom{|C|}{2}).
\end{align*}
The latter is simply an approximation of the polarization index
introduced by \citet{fal-kac-sor-szu-was:c:microscope} and the former
is a heuristic built on top of their diversity index;  we chose the
constant $5$ as
our initial experiments have shown that it
captures the same notion as their approach,
and it allows for fast
computation;
in general, computing $\empkem_i(E)$ is
intractable~\citep{fal-kac-sor-szu-was:c:microscope}, so we
approximate it using their local search approach.

We prove that for every $i$, the value of the empirical $i$-Kemeny
score of the uniformity election converges to its maximal value as the
number of candidates grows.
\begin{proposition}
\label{prop:empkem:convergence}
    For every constant $i \in \mathbb{N}$, it holds that 
    \(
        \lim_{m\rightarrow\infty}\empkem_i(\UN_{m,m!}) = \nicefrac{1}{2}\cdot m! \cdot \binom{m}{2}.
    \)
\end{proposition}

\noindent This gives the following sanity check (arrows mean
convergence as the number of candidates grows):
\begin{align*}
  D(\ID) &=       0, & A(\ID) &= 1, & P(\ID) &= 0, \\
  D(\UN) &\rightarrow 1, & A(\UN) &= 0, & P(\UN) &\rightarrow 0, \\
  D(\AN) &=       1/5, & A(\AN) &= 0, & P(\AN) &= 1. 
\end{align*}
These results are intuitive as in $\ID$ all the voters agree, $\UN$ is
most diverse, and $\AN$ is most polarized.  The DAP distance, denoted
$\hat{d}_\dap$, is a feature distance that uses $D$, $A$, and $P$ as
the features.
We see that DAP is $\ID$-consistent and $\AN$-consistent, but not
$\UN$-consistent 
(but it satisfies this property in an approximate sense%
---the distance between two $\UN$ elections with different
numbers of candidates gets smaller and smaller as the numbers of
candidates in these elections increase).
\begin{proposition}
  $\hat{d}_\dap$ is $\ID$-consistent and $\AN$-consistent, but is not $\UN$-consistent.
\end{proposition}

While our diversity feature may appear somewhat ad hoc, we believe
that it captures the intuitive notion of diversity among votes and our
results are robust to tweaking it. Indeed, in \Cref{sec:maps} we find
strong correlation between the DAP distance and the swap one (which
implicitly relies on diversity analysis
\citep{fal-kac-sor-szu-was:c:microscope}).

\subsection{Positionwise Distance Versus DAP}
\label{sec:pos-dap:diff}
Positionwise distance and DAP 
vary in several significant ways.  Foremost, the positionwise distance
deals with different numbers of candidates by, effectively, creating
their virtual copies, so that the elections it analyzes look as if
they were equal-sized (this process is hidden in computing the
Wasserstein distance and in stretching the matrices). On the other
hand, DAP identifies structural properties of the elections
(diversity, agreement, and polarization) and rescales them to a common
denominator, so that different-sized elections can be compared on
common grounds.

Second, the two distances treat top-truncated votes differently.  DAP
inherits its approach 
from the swap distance that its features are based on (the agreement
index is similar in this respect): If a voter puts some
candidates in the truncated part, then DAP assumes that he or she sees
them as equally bad. In contrast, the
positionwise distance assumes that
the voter can rank them, but chose not to report it, so the distance
assumes a uniform distribution over possible completions.

\section{Maps of Synthetic Datasets and Evaluation}
\label{sec:evaluation}

Our next goal is to understand how the positionwise and DAP distances
behave in practice. To this end, we form and analyze maps of synthetic
elections, and we evaluate how the distances between elections change
as we vary either their size or their level of truncation. First, we
describe our data.

\subsection{Datasets}\label{sec:datasets}

Our \emph{basic dataset} consists of 326 complete elections with 8 candidates
and 96 voters each, and is nearly identical to the one used by
\citet{fal-kac-sor-szu-was:c:microscope} (in particular, we chose the same 
numbers of candidates and voters as they did).
The main part of the dataset consists of elections
generated according to the impartial culture (IC), normalized
Mallows~\citep{mal:j:mallows,boe-bre-fal-nie-szu:c:compass},
Pólya-Eggenberger urn~\citep{ber:j:urn-paradox,mcc-pri-sli:j:dodgson},
and Euclidean models (see, e.g., the work
of~\citet{enelow1984spatial}).  
Under impartial culture, we draw each vote uniformly at random. The
normalized Mallows model is similar but the votes are clustered around
a given central one (the strength of this clustering is controlled by
parameter norm-$\phi \in [0,1]$,
the higher the value the less concentrated are the votes;
see the work of \citet{boe-fal-kra:c:mallows-normalization}
for a discussion of this model). The urn model generates elections
with clusters of identical votes (the larger its parameter of
contagion $\alpha \geq 0$ is, the fewer clusters there are, each
containing more votes~\citep{fal-kac-sor-szu-was:c:microscope}).
In the Euclidean models, candidates and voters are points in some
Euclidean space and voters prefer the closer candidates (we draw
the points uniformly from a unit hypercube or hypersphere of
a given dimension;
for 1-, 2-, and 3- dimensional hypercubes
we refer to the models as Interval, Square, and Cube;
for 1- and 2- dimensional hyperspheres we refer to them
as
Circle and Sphere).
The dataset also includes
elections generated using
statistical
cultures that yield
single-peaked~\citep{bla:b:polsci:committees-elections},
SPOC~\citep{pet-lac:j:spoc}, and
group-separable~\citep{ina:j:group-separable,ina:j:simple-majority}
elections.
Finally, we add $\ID$,
$\AN$, and an approximation of $\UN$ elections together with
artificial elections forming paths between these three on our maps
(which include $\ST$ election in which 
each voter ranks the same half of the candidates on top,
but otherwise the votes are chosen uniformly~at~random).

We generated the \emph{size-oriented dataset} in the same way as 
the basic one,
except that for
each culture we partitioned its elections into four groups, with
either $8$ or $16$ candidates and either $96$ or $192$ voters.
We obtain top-truncated elections from complete ones by using the following
methods:

\begin{enumerate}

\item Top-$k$ truncation 
  removes from each vote the candidates below position $k$.  Such data
  appears, e.g., in Preflib in the
  sushi dataset, where people rank their top 10 sushi types out of 100
  available ones~\citep{kam:c:sushi} (see Preflib file
    \texttt{00014-00000002.soi}).

\item Random cut truncation is parameterized by probability~$p$.  For
  each vote, we consider its candidates from top to bottom and with
  probability $1-p$ we stop the process, truncating the vote after the
  current candidate.
  Some of the political elections follow a similar pattern, e.g., UK
  Labour Party leadership
  election~\citep{ril-rya-war:www:uk-labour-party-election} (see Preflib
    file \texttt{00030-00000001.soi}).

  \item Random drop truncation moves each candidate in each vote to
    the truncated part, independently, with probability $p$. This
    imitates sport elections---such as those for Formula~1---where
    each player fails to finish a given competition with some
    probability.

\end{enumerate}
To form the \emph{comprehensive dataset},
we took the size-oriented
one and for each group of elections (of a given size, generated using
a given statistical culture) we left the first half of the elections
in the group intact, we applied the top-$k$ truncation to the next
quarter of them, and we applied random cut truncation to the last
quarter. We chose the truncation parameters so that, in expectation,
each voter ranked half of the candidates.  We generated the
\emph{truncation-oriented dataset} in the same way, but starting from the
basic dataset.
We also generated a \emph{random drop dataset},
but we omitted it from the experiments as its elections have very
different nature
(we analyze them in \Cref{app:truncation:analysis}).

\begin{figure*}[t]
  \centering
  \hfill
  \begin{subfigure}[b]{0.35\linewidth}
    \includegraphics[width=\linewidth]{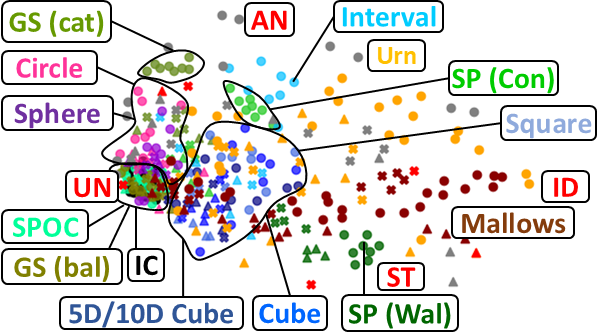}%
    \caption{Positionwise Distance Map.}
    \label{fig:map:trun1-poswise}
  \end{subfigure}
  \hfill
  \begin{subfigure}[b]{0.35\linewidth}
    \includegraphics[width=\linewidth, trim={0cm 0cm 0cm 0cm}, clip]{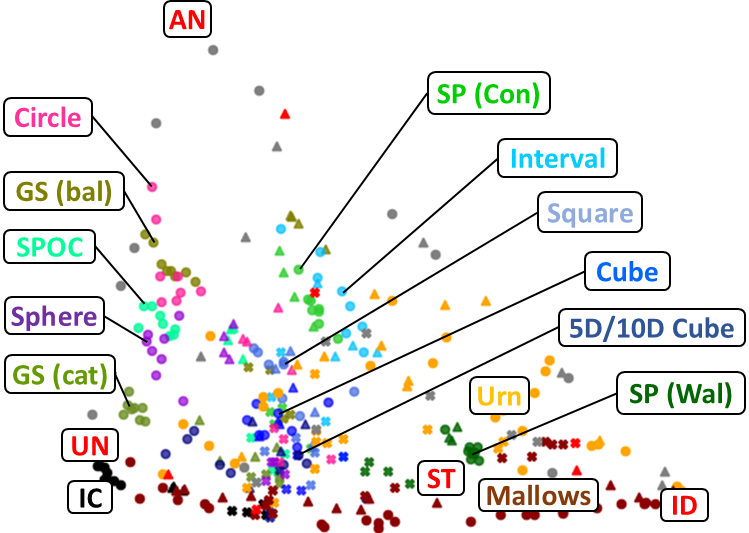}%
    \caption{DAP Distance Map.}
    \label{fig:map:trun1-dap-ape}
  \end{subfigure}
  \hfill
  \begin{subfigure}[b]{0.01\linewidth}
  \end{subfigure}
  
  \caption{Maps of elections created using the positionwise and
    DAP distances, for the 
    truncation-oriented datasets.  Top-$k$ truncated elections are marked with
    triangles, random-cut truncated ones with crosses, and complete
    ones with circles.}
  \label{fig:big-map-experiment}
\end{figure*}

\begin{figure*}[t]
    \centering
    \includegraphics[width=0.97\linewidth]{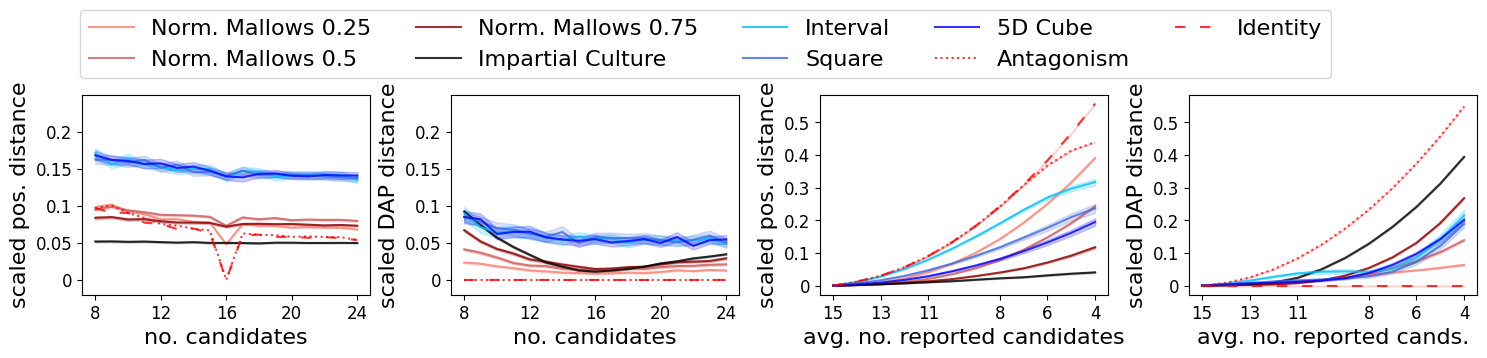}
    \begin{minipage}{0.48\linewidth}
      \centering (a) Distances between complete elections of different sizes.
    \end{minipage}\quad
    \begin{minipage}{0.48\linewidth}
      \centering (b) Distances between elections and their top-$k$ truncations.
    \end{minipage}
    
    \caption{Plots (a) show average distance from size-16 elections to
      different-sized, complete elections from the same culture, as a
      fraction of the maximum distance in our dataset (positionwise distance on the left, DAP distance on the right). Plots (b) show  average distance
      from a complete election to its top-$k$ truncation, as a
      fraction of the maximum distance in our dataset (positionwise distance on the left, DAP distance on the right).}
    \label{fig:sizes-truncs}
  
\end{figure*}

\subsection{Maps of Elections}\label{sec:maps}

We show maps for the truncation-based dataset
obtained using our two distances in
\Cref{fig:big-map-experiment}, with the KK
embedding~\citep{kam-kaw:j:embedding} used for the positionwise distance,
and MDS~\citep{kru:j:mds} used for DAP.
The maps for size-based and comprehensive
datasets can be found in \Cref{app:maps-of-elections}.

The map obtained using the positionwise distance
is similar
to that of the basic dataset (shown in \Cref{fig:map:poswise})
and the map obtained using the DAP distance resembles
the map of the basic dataset obtained using the
swap distance (\Cref{fig:map:swap}),
albeit with some level of degradation.
For the positionwise map, the truncated
elections tend to be closer to $\UN$ than their complete
counterparts
and for DAP we see a cluster of random-cut elections one-third of a
way between $\UN$ and $\ID$ elections.
Unfortunately, if we used random drop truncation (with dropping
probability of $0.5$) then all thus-generated elections would form a
single cluster in the vicinity of $\UN$. One explanation for this is
that random drop affects the internal structure of elections very
strongly. For example, in the $\AN_8$ election only two candidates are
ever ranked first, but with random drop
truncation
every candidate has a nonnegligible
probability of being ranked first.
Thus, analyzing random-drop elections is particularly
difficult.
Still, the maps indicate that, on the high level, both distances
give intuitive, reasonable results.

\subsection{Varying Election Sizes  Truncation Level}\label{sec:quantitative-exp}

Next, we evaluate the robustness of our distances in a more
quantitative way.  First, we analyze their ability to recognize
similar elections with different numbers of candidates.  To this end,
for each of several statistical cultures (IC, normalized Mallows model
with $\normphi \in \{0.25,0.5,0.75\}$, and 1D/2D/5D Euclidean models
with points distributed uniformly on unit hypercubes) and for each
integer~$m$ between $8$ and $24$, we generated 100 pairs of elections,
each with 192 voters, where one election in the pair had $m$
candidates and the other one had $16$, and we computed their average
distance (normalized by the largest distance that occurred in the
datasets from \Cref{sec:datasets}, i.e., an approximate diameter of
the election space; consequently, the results for positionwise and DAP
are on the same scale).  We report the results in \Cref{fig:sizes-truncs}(a),
the shaded areas show 95\% confidence intervals (we also included
$\ID$ and $\AN$ elections in the plot; IC can be seen as an
approximation of $\UN$).
Ideally, we would like our plots to consist of flat lines, close to
zero. This would mean that a given distance can recognize structural
similarities between elections generated in the same way, even if
these elections have different numbers of candidates. Hence, in our
view DAP performs better as its scaled values are significantly
smaller, even if the results for Mallows are less flat (as the Mallows
and IC lines are increasing, the reader may worry what happens for even more candidates; in short, they  increase slowly, with IC reaching about $10\%$ of the diameter
for the case of 100 candidates).
It is reassuring that for DAP the plots for $\ID$ and $\AN$
are flat  at value zero (as DAP is
$\ID$- and $\AN$-consistent), and for positionwise the plot for IC is
flat and close to zero (as positionwise
is $\UN$-consistent and IC elections approximate
$\UN$).

To analyze the influence of truncation,
for each of the cultures from the previous experiment we
generated 25 elections
 with 16 candidates and 192 voters, and
applied each of our
truncation methods to each of them
changing the parameter
to obtain varying levels of vote completeness.
For each
truncated election, we computed its distance from the original,
complete one (and normalized it as previously). 
In \Cref{fig:sizes-truncs}(b) we plot the average of these values, for
the case of top-$k$ truncation;
shaded area shows 95\% confidence interval.  We see
that if the voters rank at least half of the candidates, then the
average distance of the truncated election from its complete variant
is (a)~less than $20\%$ of the diameter for the positionwise distance
(indeed, even less than $10\%$ for most of the cultures), and (b)~less
than $5\%$ of the diamater for DAP (except for $\AN$ and IC elections, where it
is  $\leq 30\%$ and  $\leq 15\%$ of the diameter, respectively). Thus, both distances
handle truncated data well, but DAP has an advantage (however, for the
other truncation types DAP and positionwise perform similarly to each
other).
Overall, the experiments point to  DAP.

\section{Map of Preflib}\label{sec:map-of-preflib}
Last but not least, we present the \emph{Map of Preflib}, obtained using the
DAP distance
(we offer more in-depth analysis in \Cref{app:preflib:all}).
\citet{boe-bre-fal-nie-szu:c:compass} already tried putting Preflib
elections on the map, but 
their approach was limited to 
elections with
$10$ candidates and $100$ complete votes, obtained by involved
preprocessing.
We use raw, unprocessed elections from Preflib, with the number of candidates ranging from $3$ to
$\approx\! 2~600$ and the number of voters from $4$ to $\approx\! 64~000$.
We also use more elections, as Preflib was extended since their work.

We show our map in \Cref{fig:map:preflib}: Each black dot represents a
single election from Preflib, while large pale discs represent
elections generated synthetically
(much more detailed analysis is available in the extended version of the paper).
We find that most of the Preflib elections fall between Euclidean ones
(with uniform distribution of candidate and voter points
inside a unit hypercube of between 2 and 10 dimensions),
Mallows elections (with $\normphi$ values in the range $[0.2,0.7]$),
and urn elections (with contagion parameter in the range $[0.2,0.5]$,
or $[0.2,2]$ if one wants to include elections with large
clusters of identical votes; similar elections happen in Preflib, but
are more rare).
This
motivates the use of these models and 
parameter ranges
to generate realistic-looking 
data using only the few 
best-known statistical cultures.  Some
previous papers, including those of \citet{boe-bre-fal-nie-szu:c:compass} and \citet{fal-kac-sor-szu-was:c:microscope}, argued 
that the urn model does not generate 
realistic elections but our results counter this.
That said, we also see an area with many Preflib elections,
but no elections from our statistical cultures (black dots over white
background). It would be interesting to find statistical cultures that
do cover this space.

Next, let us anlyze the locations of the Preflib elections with
respect to $\ID$, $\UN$, and $\AN$.  Foremost, we see that the area
near $\AN$ is empty and, hence, Preflib elections are not strongly
polarized (the same effect was visible in the map of
\citet{boe-bre-fal-nie-szu:c:compass}). Some elections that (weakly)
stand out in this respect include those in the WebSearch dataset
(but they include only $4$ votes each and, hence, are very particular),
some Netflix elections (where each voter ranks the same $3$ or $4$
movies, so some level of disagreement is expected) and some APA
elections (which involve 5 candidates for the president of the
American Psychological Association, with some tension between
``academics'' and ``clinicians''). Elections from the (Figure)
Skating, Formula~1, and Countries datasets represent the other extreme
and are located between $\ID$ and $\UN$, close to the Mallows
elections.  These locations are quite natural: Indeed, we expect the
judges in figure skating, who evaluate the same performances, conducted on the
same day, to be correlated, and we expect Formula~1 races that happen
over the course of a season to be more varied.

\begin{figure}[t]
  \centering
    \includegraphics[width=0.725\linewidth]{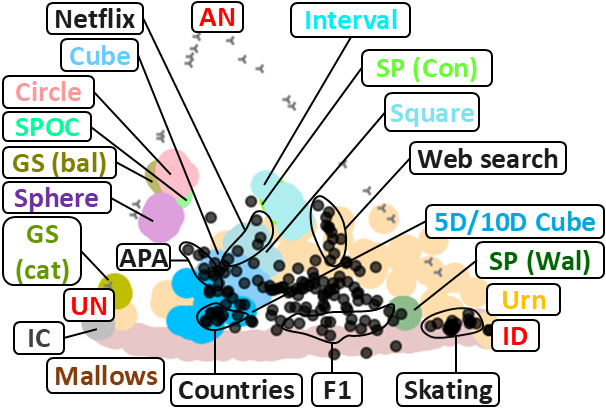}%
  \caption{
    Map of Preflib elections (black dots)
    in addition to the synthetic ones (large pale discs),
    obtained using the DAP distance.}
  \label{fig:map:preflib}
\end{figure}

\section{Summary}

We found that the DAP distance is an interpretable, scalable way of
assessing similarity between different-sized/top-truncated elections,
suitable to form a map of (a fragment of) Preflib. By analyzing this
map, we were able to draw a number of conclusions, including the
parameters for Mallows, urn, and Euclidean models that yield realistic
elections.

Nevertheless, we believe that there is no one-fit-all method for creating maps, and different applications may call for different approaches.
Thus, considering different distances,
including but not limited to
feature distances with different set of features,
and developing arguments for their usefulness
might be a fruitful direction for future research.

\section*{Acknowledgements}
This project has received funding from the European Research Council (ERC) under
the European Union’s Horizon 2020 research and innovation programme (grant
agreement No 101002854), from the French government under management of
Agence Nationale de la Recherche as part of the "Investissements d'avenir"
program, reference ANR-19-P3IA-0001 (PRAIRIE 3IA Institute), and from the European Union under the project Robotics and advanced industrial production (reg. no. CZ.02.01.01/00/22\_008/0004590). 
J.\@ Mertlová acknowledges the additional support of the Faculty of Information Technology of the Czech Technical University in Prague via the Výzkumné léto (VýLeT) project.
S.\@ Szufa was supported by the Foundation for Polish Science (FNP). 
T.\@ Wąs was partially supported by EPSRC under grant EP/X038548/.
\begin{center}
  \includegraphics[width=3cm]{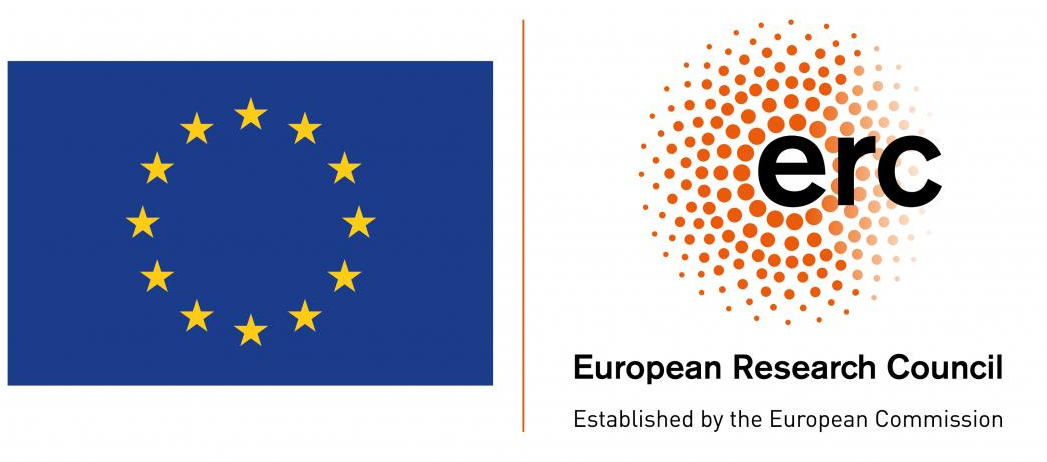}
  \includegraphics[width=3cm]{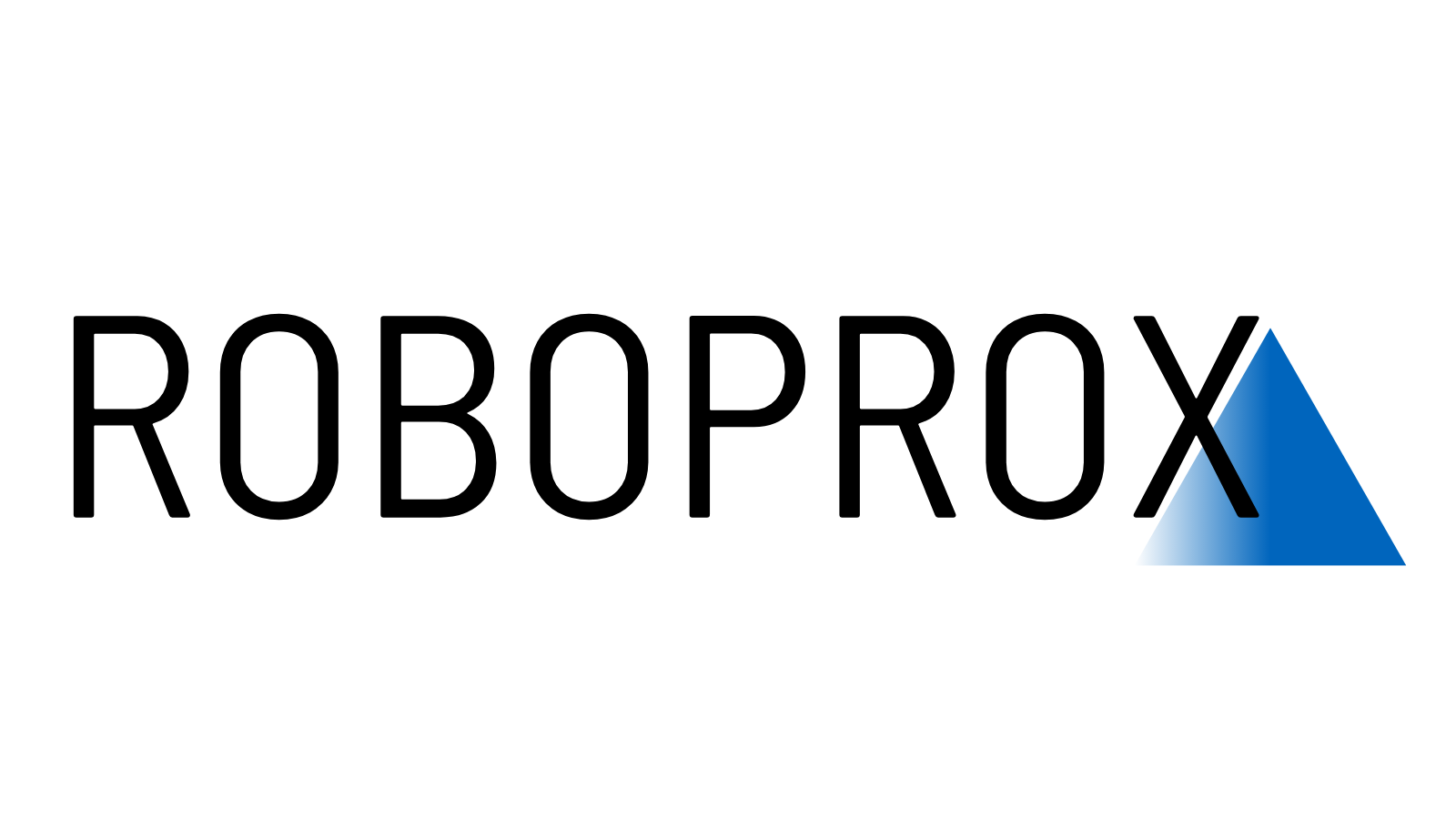}
\end{center}

\bibliography{bib,references_}

\clearpage

\appendix 

\section{Relation Between EMD and Wasserstein Distances}\label{sec:emd-vs-W}
The original positionwise distance, as defined by
\citetapp{szu-fal-sko-sli-tal:c:map} used earth mover's distance (EMD) between
vectors and not Wasserstein's, as we do. In this section we argue that
the difference between the two distances is small and bounded (for the
case of vectors of the same dimension).

First, let us recall definitions of the EMD and $\ell_1$ distances.
Given two vectors,\ $\vec{a} = [a_1, \ldots, a_m]$ and
$\vec{b} = [b_1, \ldots, b_m]$, their $\ell_1$ distance is:
\[
  \ell_1(\vec{a},\vec{b}) = |a_1-b_1| + |a_2-b_2| + \cdots + |a_m-b_m|.
\]
If, additionally, their entries are nonnegative and sum up to $1$,
then their EMD distance, denoted $\emd(\vec{a},\vec{b})$, is the
smallest total cost of transforming $\vec{a}$ into $\vec{b}$ through
operations of the following form: Given $i, j \in [m]$ and a number
$\delta \leq a_i$, we can move value $\delta$ from position $i$ in
$\vec{a}$ to position $j$ at cost $\delta \cdot |i-j|$. In other
words, we view vector $\vec{a}$ as a collection of~$m$ buckets, where
each bucket $i$ contains $a_i$ amount of sand. We can move the sand
between the buckets, but each such operation costs proportionally to
the distance between the buckets and the amount of sand moved. It is
also well-known that we can express $\emd(\vec{a}, \vec{b})$ using the
$\ell_1$ distance and the notation where for each $i \in [m]$,
$\hat{a}_i$ and $\hat{b}_i$ mean $a_1 + \cdots + a_i$ and
$b_1 + \cdots + b_i$, respectively~\citepapp{rub-tom-gui:j:emd}:
\begin{align*}
  \emd(\vec{a},\vec{b}) &= |\hat{a}_1-\hat{b}_1| + \cdots + |\hat{a}_m - \hat{b}_m|.
\end{align*}

\begin{proposition}\label{prop:emd-vs-w}
  For each two vectors, $\vec{a} = [a_1, \ldots, a_m]$ and $\vec{b} = [b_1, \ldots, b_m]$,
  with nonnegative entries that sum up to $1$, it holds that:
  \begin{align*}
    \emd(\vec{a},\vec{b}) - 1 \leq m\cdot W(\vec{a},\vec{b}) \leq \emd(\vec{a},\vec{b})
  \end{align*}  
\end{proposition}
\begin{proof}
  Let $\vec{a}$ and $\vec{b}$ be as in the statement of the
  proposition. 
  Following the definition of their Wasserstein's distance from
  \Cref{sec:prelim}, we consider their associated functions
  $a(\cdot)$, $b(\cdot)$, $A(\cdot)$, and
  $B(\cdot)$. In particular, we observe that for each $i \in
  [m]$ and each $x \in [\frac{i-1}{m},\frac{i}{m}]$ we have:
  \begin{align*}
    a(x) = a_im, && b(x) = b_im.
  \end{align*}
  Since $A(x) = \textstyle \int_0^x a(y) dy$ and $B(x) = \textstyle
  \int_0^x b(y) dy$, for each $i \in [m]$ and each $x \in
  [\frac{i-1}{m},\frac{i}{m}]$ we also have:
  \begin{align*}
    A(x) &= \textstyle \hat{a}_{i-1} +\left(x-\nicefrac{i-1}{m} \right)a_im, \text{ and}\\
    B(x) &= \textstyle \hat{b}_{i-1} +\left(x-\nicefrac{i-1}{m} \right)b_im, 
  \end{align*}
  which, after basic calculations, gives:
  \begin{align*}
    A(x) &= \textstyle \hat{a}_{i} -\left(\nicefrac{i}{m}-x \right)a_im, \text{ and}\\
    B(x) &= \textstyle \hat{b}_{i} -\left(\nicefrac{i}{m}-x \right)b_im.
  \end{align*}
  Now we are ready to calculate $W(\vec{a},\vec{b})$. We have:
  \begin{align}
    \nonumber
    W&(\vec{a},\vec{b})  = \int_0^1|A(x)-B(x)|dx \\
     &= \sum_{i=1}^m \left( \int_{\frac{i-1}{m}}^{\frac{i}{m}}|A(x)-B(x)|dx \right)
       \label{eq:emd-w}
  \end{align}
  For each $i \in [m]$, we express the integral under the sum as:
  \begin{align}
    \int_{\frac{i-1}{m}}^{\frac{i}{m}}\big|&\hat{a}_i - \hat{b}_i - \left(\nicefrac{i}{m}-x\right)(b_i-a_i)m\big|dx,
                                             \label{eq:emd-w2}
  \end{align}
  which is greater or equal to the following expression (see
  explanations below):
  \begin{align*}
    \textstyle\int_{\frac{i-1}{m}}^{\frac{i}{m}}&|\hat{a}_i - \hat{b}_i|dx - \textstyle\int_{\frac{i-1}{m}}^{\frac{i}{m}} |(\nicefrac{i}{m}-x)(b_i-a_i)m|dx, \\
    = &\textstyle\frac{1}{m}|\hat{a}_i - \hat{b}_i| - \textstyle\frac{1}{2m}|b_i-a_i|. 
  \end{align*}
  The fact that the first expression in the above equality is smaller
  or equal than \eqref{eq:emd-w2} follows from the fact that for each
  $\alpha, \beta \in \reals$, $|\alpha + \beta| \geq
  |\alpha|-|\beta|$, and the next equality follows by computing the
  integrals (in particular, $\int_{\frac{i-1}{m}}^{\frac{i}{m}}
  |(\nicefrac{i}{m}-x)(b_i-a_i)m|dx$ is the area of a triangle with
  height $\nicefrac{1}{m}$ and base equal to $|b_i-a_i|$).

  Altogether, after substituting these calculations into
  \eqref{eq:emd-w}, we obtain that
  $W(\vec{a},\vec{b})$ is greater or equal to:
  \begin{align*}
    &\textstyle
      \sum_{i=1}^m \frac{1}{m}|\hat{a}_i - \hat{b}_i| - \frac{|a_i-b_i|}{2m} 
      =  \nicefrac{1}{m}\left( \emd(\vec{a},\vec{b}) - \textstyle \frac{1}{2}\ell_1(\vec{a},\vec{b})\right).
  \end{align*}
  By observing that $\ell_1(\vec{a},\vec{b}) \leq
  2$, we have the first inequality from the statement of the
  proposition.

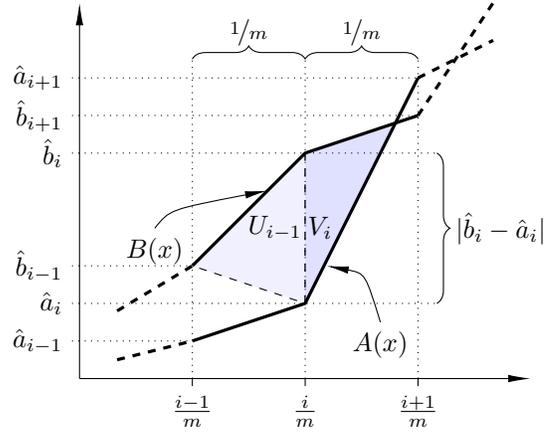
\begin{figure}
  \centering
  \begin{tikzpicture}
    \filldraw[fill=blue!6!white, dashed] (3,1) -- (3,3) -- (1.5,1.5) -- (3,1);
    \filldraw[fill=blue!12!white, dashed] (3,1) -- (3,3) -- (4.2,3.4) -- (3,1);
    \draw[dotted] (1.5,0) -- (1.5,4.25);
    \draw[dotted] (3,0) -- (3,4.25);
    \draw[dotted] (4.5,0) -- (4.5,4.25);
    
    \draw[->] (0,0) -- (6,0);
    \draw[->] (0,0) -- (0,5);
    \draw (1.5,0) -- (1.5,-0.1);
    \draw (3,0) -- (3,-0.1);
    \draw (4.5,0) -- (4.5,-0.1);
    \node[anchor=north] at (1.5,-0.1) {$\frac{i-1}{m}$};
    \node[anchor=north] at (3,-0.1) {$\frac{i}{m}$};
    \node[anchor=north] at (4.5,-0.1) {$\frac{i+1}{m}$};
    \draw[very thick, dashed] (0.5,0.25) -- (1.5,0.5);
    \draw[very thick, dashed] (4.5,4) -- (5.5,4.5);
    \draw[very thick] (1.5,0.5) -- (3,1) -- (4.5,4);

    \draw[very thick, dashed] (0.5,0.9) -- (1.5,1.5);
    \draw[very thick, dashed] (4.5,3.5) -- (5.5,5);
    \draw[very thick] (1.5,1.5) -- (3,3) -- (4.5,3.5);

    \draw[dotted] (-0.1,0.5) -- (1.5,0.5);
    \draw[dotted] (-0.1,1) -- (4.75,1);
    \draw[dotted] (-0.1,4) -- (4.5,4);
    \draw[dotted] (-0.1,1.5) -- (1.5,1.5);
    \draw[dotted] (-0.1,3) -- (4.75,3);
    \draw[dotted] (-0.1,3.5) -- (4.5,3.5);

    \node[anchor=east] at (-0.1,0.5) {$\hat{a}_{i-1}$};
    \node[anchor=east] at (-0.1,1) {$\hat{a}_{i}$};
    \node[anchor=east] at (-0.1,4) {$\hat{a}_{i+1}$};
    \node[anchor=east] at (-0.1,1.5) {$\hat{b}_{i-1}$};
    \node[anchor=east] at (-0.1,3) {$\hat{b}_{i}$};
    \node[anchor=east] at (-0.1,3.5) {$\hat{b}_{i+1}$};

    \node[anchor=center] at (2.6,2) {$U_{i-1}$};
    \node[anchor=center] at (3.2,2) {$V_i$};

    \draw[decorate,decoration={brace,amplitude=5pt,mirror}]
  (4.75,1) -- (4.75,3) node[midway,xshift=0.85cm]{$|\hat{b}_i-\hat{a}_i|$};
    \draw[decorate,decoration={brace,amplitude=5pt}]
  (1.5,4.25) -- (3,4.25) node[midway,yshift=10pt]{$\nicefrac{1}{m}$};
    \draw[decorate,decoration={brace,amplitude=5pt}]
    (3,4.25) -- (4.5,4.25) node[midway,yshift=10pt]{$\nicefrac{1}{m}$};

    \draw[->] (1,2) .. controls (1.2,2.25) .. (2.5,2.5);
    \node[anchor=north] at (1,2) {$B(x)$};

    \draw[->] (4,0.75) .. controls (3.8,1.4) .. (3.28,1.5);
    \node[anchor=north] at (4,0.75) {$A(x)$};

  \end{tikzpicture}
  \caption{Illustration of the proof that $W(\vec{a},\vec{b}) \leq \frac{1}{m}\emd(\vec{a},\vec{b})$.}
  \label{fig:emd-wasserstein}
  \end{figure}

  Next, we show that
  $W(\vec{a},\vec{b}) \leq
  \frac{1}{m}\emd(\vec{a},\vec{b})$. Geometric interpretation of
  $W(\vec{a},\vec{b})$ is that it is the area between functions $A(x)$
  and $B(x)$, for $x \in [0,1]$. From the preceding arguments, we know
  that for each $i \in [m]$, $A(\nicefrac{i}{m}) = \hat{a}_i$ and
  $B(\nicefrac{i}{m}) = \hat{b}_i$. Further, both functions are linear
  within each interval $[\nicefrac{i-1}{m},\nicefrac{i}{m}]$. We
  partition the area between functions $A(x)$ and $B(x)$ into $2m$
  triangles as follows (see \Cref{fig:emd-wasserstein} for
  illustration): For each $i \in [m]$, we form triangle $U_{i-1}$
  connecting points $(\frac{i}{m},\hat{a}_i)$,
  $(\frac{i}{m},\hat{b}_i)$ and either
  $(\frac{i-1}{m},\hat{b}_{i-1})$, if $A(x)$ and $B(x)$ do not
  intersect in the interval $[\frac{i-1}{m},\frac{i}{m}]$, or their
  intersection point in that interval, if they do. Similarly, we form
  triangle $V_i$ connecting points $(\frac{i}{m},\hat{a}_i)$,
  $(\frac{i}{m},\hat{b}_i)$ and either
  $(\frac{i+1}{m},\hat{a}_{i-1})$, if $A(x)$ and $B(x)$ do not
  intersect in the interval $[\frac{i}{m},\frac{i+1}{m}]$, or their
  intersection point in that interval, if they do (we take $V_m$ to be
  an empty triangle). Note that together the area of these triangles
  is equal to $W(\vec{a},\vec{b})$. Further, one can easily verify
  that for each $i \in [m]$, the area of each of the triangles
  $U_{i-1}$ and $V_{i}$ is bounded by
  $\frac{1}{2}\cdot\frac{1}{m}|\hat{a}_i-\hat{b}_i|$. Consequently, we
  have
  $W(\vec{a},\vec{b}) \leq \sum_{i=1}^m
  \frac{1}{m}|\hat{a}_i-\hat{b}_i| =
  \frac{1}{m}\emd(\vec{a},\vec{b})$.
\end{proof}

By a careful analysis of the second part of the above proof, we also
obtain the following corollary.
\begin{corollary}\label{cor:emd-vs-w}
  For each two vectors, $\vec{a} = [a_1, \ldots, a_m]$ and
  $\vec{b} = [b_1, \ldots, b_m]$, with nonnegative entries that sum up
  to $1$, it holds that:
  \begin{align*}
    \textstyle \frac{1}{2}\emd(\vec{a},\vec{b})  \leq m\cdot W(\vec{a},\vec{b}) \leq \emd(\vec{a},\vec{b})
  \end{align*}  
\end{corollary}
In essence, it follows by observing that if we replaced function
$A(x)$ with $A'(x) = \min(A(x),B(x))$ and function $B(x)$ with
$B'(x) = \max(A(x),B(x))$ then we would have:
\[
  \emd(\vec{a},\vec{b})  = \int_0^1|A'(x)-B'(x)|dx \\
\]
This integral is at most twice as large as the
$\int_0^1|A(x)-B(x)|dx$.  Indeed, if for a given $i \in [m]$ and all
$x \in [\nicefrac{i-1}{m},\nicefrac{i}{m}]$ we have that
$A(x) \leq B(x)$, then the two integrals are equal for this range of
$x$. If, on the other hand, the graphs of functions $A(x)$ and $B(x)$
cross, then
$\int_{\nicefrac{i-1}{m}}^{\nicefrac{i}{m}}|A'(x)-B'(x)|dx \leq
2\int_{\nicefrac{i-1}{m}}^{\nicefrac{i}{m}}|A(x)-B(x)|dx$. This
follows because the area of a trapezoid is at most twice as large as
the joint area of the two triangles formed by its bases and diagonals.

Jointly, \Cref{prop:emd-vs-w} and \Cref{cor:emd-vs-w} say that the
Wasserstain distance is similar to EMD both when the input vectors are
close to each other (in which case \Cref{cor:emd-vs-w} gives a more
accurate bound) and when they are far away (in which case
\Cref{prop:emd-vs-w} kicks in).

Finally, we observe that the bound in \Cref{cor:emd-vs-w} is tight,
but the one in \Cref{prop:emd-vs-w} can potentially be improved. To
this end, let us fix a positive integer $m$ and consider the following
two vectors of dimension $2m+1$:
\begin{align*}
  \vec{a} &= \textstyle [\frac{1}{2m},0,\frac{1}{m},0,\frac{1}{m}, \ldots,0,\frac{1}{m},0,\frac{1}{2m}] \\
  \vec{b} &= \textstyle [0,\frac{1}{m},0,\frac{1}{m}, \ldots,0,\frac{1}{m},0]
\end{align*}
We can see that $\emd(\vec{a},\vec{b}) = 1$ (in essence, to obtain
$\vec{b}$ from $\vec{a}$, both $\frac{1}{2m}$ entries need to move to
their neighbors in full, whereas each $\frac{1}{m}$ entry needs to
move half of its value to the left and half of its value to the
right). Careful computation also shows that
$(2m+1) \cdot W(\vec{a},\vec{b})$ is equal to:
\begin{align*}
  (2m+1) & \cdot \left( \frac{1}{2m(2m+1)} + (2m-1)\frac{1}{4m(2m+1)} \right) \\
         & = \frac{1}{2m} + (2m-1)\frac{1}{4m} = \frac{2m+1}{4m} = \frac{1}{2}+\frac{1}{4m}.
\end{align*}
As $m$ goes to infinity, this values approaches $\frac{1}{2}$, i.e.,
half of the value of the EMD distance.

\section{Statistical Cultures and Our Dataset}\label{app:cultures}

We first describe the statistical cultures that we use and then we
give the details of our synthetic datasets.

\subsection{Statistical Cultures}

All the statistical cultures that we consider generate complete
elections. Hence, whenever below we speak of a vote, we mean a
complete one.

\paragraph{Impartial Culture.}
Under impartial culture (IC) we generate elections vote-by-vote,
drawing the preferences of each voter uniformly at random from the set
of all possible ones.

\paragraph{(Normalized) Mallows Model.}
Mallows model is similar to IC in that we also generate votes
one-by-one, but instead of drawing them uniformly at random, we use a
distribution clustered around a given central vote. Formally, if $u$ is
the \emph{central vote} and $\phi \in [0,1]$ is the so-called
\emph{parameter of dispersion}, then the probability of generating
vote $v$ is proportional to
$\phi^{\swap(u,v)}$~\citepapp{mal:j:mallows}. However, instead of
using the parameter $\phi$ directly, we use its normalized variant
called $\normphi \in [0,1]$ and introduced by
\citetapp{boe-bre-fal-nie-szu:c:compass}. Formally, given a value of
$\normphi$ and the number $m$ of candidates that we consider, we use a
value of $\phi$ such that the expected swap distance between the votes
that the Mallows model generates and the central one is
equal to:
\[
  \normphi \cdot \frac{m(m-1)}{4}.
\]
Intuitively, $\frac{m(m-1)}{4}$ is half of the maximum swap distance
between two preference orders over $m$ candidates, so setting
$\normphi = 1$ leads to generating maximally diverse votes (indeed, in
this case the model is equivalent to IC), whereas setting $\normphi=0$
requires all generated votes to be identical to the central one.
\citetapp{boe-bre-fal-nie-szu:c:compass} and
\citetapp{boe-fal-kra:c:mallows-normalization} argue how choosing the
$\normphi$ values in between $0$ and $1$ leads to smooth transition
between these two extremes, independent of the number of candidates
(without the normalization the transition becomes more and more abrupt
as the number of candidates grows).

\paragraph{Polya-Eggenberger Urn Model.} The urn model uses the
\emph{parameter of contagion} $\alpha \geq 0$. The model was
introduced by \citetapp{ber:j:urn-paradox}, but the parameterization
that we use is due to \citetapp{mcc-sli:j:similarity-rules}.  To
generate an election with $n$ voters who express preferences over $m$
candidates, we proceed as follows. First, we form an urn that contains
a single copy of every possible preference order (i.e., altogether
$m!$ votes).  Then, to generate a vote, we draw a single preference
order from the urn, include its copy in the election, and return it to
the urn, together with $\alpha m!$ of its copies. We repeat this process $n$
times.

For $\alpha = 0$ the urn model is equivalent to IC, whereas for
$\alpha = 1$ the probability that the second vote is identical to the
first one is $\frac{1}{2}$.  Generally, the larger is the parameter
$\alpha$, the fewer clusters of identical votes, but each with more
votes, are included in the generated elections
(\citetapp{fal-kac-sor-szu-was:c:microscope} give an upper bound on
the expected number of clusters depending on $\alpha$).

\paragraph{Euclidean Models.}
In the Euclidean models we assume that each candidate and each voter
is represented as a point in some Euclidean space, and the voters
prefer those candidates whose points are closer to theirs. Formally,
if $v$ is a voter, $a,b$ are candidates, and the point of $v$ is
closer to that of $a$ than to that of $b$, then $v$ prefers $a$ to $b$
(we disregard the possibility of ties as, due to our way of generating
the candidate and voter points, they occurr with negligible
probability\footnote{Mathematically, the probability of a tie is zero,
  but due to finite nature of computation on real-life computers, it
  is slightly larger.}).

We consider two main ways of generating the points of the candidates
and voters: Either we draw them uniformly at random from some
$d$-dimensional hypercube, or from some $d$-dimensional hypersphere.
Regarding the hypercube models, we consider the Interval model (points
generated from a unit interval), the Square model (points generated
from a square), the Cube model (points generated from a 3D cube), and
5D- and 10D Cube models (points generated from 5D and 10D cubes,
respectively). Regarding the hypersphere models, we consider the Circle
model (with the points generated from a circle in the 2D space) and
the Sphere model (with the points generated from a 3D sphere).

\paragraph{Single-Peaked Elections.}
Let $E = (C,V)$ be an election, where $C = \{c_1, \ldots, c_m\}$ and
$V = (v_1, \ldots, v_n)$. Furher, let $\triangleleft$ be strict linear
order over $C$, referred to as the societal axis. We say that $E$ is
\emph{single-peaked with respect to $\triangleleft$} if for every vote
$v_i$ and every $t \in [m]$ the top $t$ candidates in $v_i$'s
preference order form an interval with respect to
$\triangleleft$; $E$ is \emph{single-peaked} if there is a
societal axis with respect to which it is single-peaked. The notion of
being \emph{single-peaked on a circle (SPOC)} is defined analogously,
except that for each voter $v_i$ and each $t \in [m]$ the top $t$
candidates according to $v_i$ either form an inteval with respect to
the societal axis or a complement of an inteval.  Single-peakedness
was introduced by \citetapp{bla:b:polsci:committees-elections} whereas
single-peakedness on a circle was studied by \citetapp{pet-lac:j:spoc}.

Intuitively, single-peaked elections capture settings where there is
some objective, single-dimensional criterion (such as the political
left-right spectrum, or a spectrum of possible temperatures in the
room) with respect to which the candidates can be ordered, each voter
has a favorite option, and the further a candidate is from this
favorite one, the lower it is ranked by the respective voter. SPOC
elections, on the other hand, capture such settings as voting on a
time for a conference call when different participants are in
different time zones.

We consider two models of generating single-peaked elections (in each
model we draw the societal axis $\triangleleft$ uniformly at random
and then generate the votes one-by-one, independently):
\begin{enumerate}
\item In the Conitzer model~\citepapp{con:j:eliciting-singlepeaked},
  denoted SP (Con), the process of generating a vote is as
  follows. First, we choose the top candidate uniformly at
  random. Then we perform $m-1$ steps where in each step we extend the
  vote with the next-best candidate, either located directly to the
  left (with respect to the societal axis) of those already ranked, or
  directly to the right (the choice is made with probability
  $\nicefrac{1}{2}$, unless there are no more candidates on one side,
  in which case the candidate from the other side is selected
  deterministically).
\item In the Walsh model~\citepapp{wal:t:generate-sp}, denoted SP
  (Wal), each single-peaked vote is generated with equal probability,
  uniformly at random.
\end{enumerate}

For the case of SPOC elections, we generate the votes uniformly at
random (indeed, in this case the naturally adapted variant of the
Conitzer model coincides with that of Walsh).

\paragraph{Group-Separable Elections.}
Group-separable elections were introduced by
Inada~\citepapp{ina:j:group-separable,ina:j:simple-majority}, but in
our discussion we will use the tree-based definition presented by
\citetapp{kar:j:group-separable}. Let $C = \{c_1, \ldots, c_m\}$ be a
set of candidates and let $\calT$ be rooted, ordered tree with $m$
leaves, where each leaf is labeled with a distinct candidate. For
simplicity, we say that each internal node of $\calT$ orders its
children from left to right. A \emph{frontier} of $\calT$ is the
preference order over $C$ that ranks the candidates in the order in
which we would visit them if we traversed $\calT$ using DFS, starting
from the root and visiting each node's children in the order from left
to right (inuitively, if we drew such a tree on paper, then we would
obtain the frontier by reading the candidates in the leaves from left
to right).  A vote $v$ is compatible with tree $\calT$ if it can be
obtained as its frontier by reversing the order in which (some of) the
internal nodes list their children. An election $E = (C,V)$ is
\emph{group-separable} is there exists a rooted, ordered tree $\calT$,
where each leaf is labeled with a unique member of $C$, such that each
vote in $V$ is compatible with $\calT$.

We are interested in two subclasses of group-separable elections:
\begin{enumerate}
\item We say that an election is \emph{balanced group-separable},
  denoted GS (bal), if it is group-separable with respect to some
  complete binary tree (i.e., a binary tree where all levels, except
  possibly for the last one, are filled completely).
\item We say that an election is \emph{caterpillar group-separable},
  denoted GS (cat), if it is group-separable with respect to some
  caterpillar binary tree (i.e., a binary tree where each node either
  is a leaf or one of its children is a leaf).
\end{enumerate}

To generate a balanced (caterpillar) group-separable election, we
start with an arbitrary balanced (caterpillar) tree that has its
leaves labeled with the candidates that we are interested in and we
generate votes one-by-one as follows: For each node of the tree, we
reverse the order of its children with probability $\nicefrac{1}{2}$
and then we add the frontier of the tree to the election.

\begin{table}
  \centering
  \caption{\label{tab:app:basic}Numbers of elections generated
    according to our statistical cultures in the basic datatset. For
    the normalized Mallows elections, we choose $\normphi$ uniformly
    at random from interval $[0,1]$. For the urn model, we choose the
    $\alpha$ parameter according to the Gamma distribution with the
    shape parameter set to $0.8$ and the scale parameter set to $1$
    (both parameterizations are in sync with previous datasets used in
    the maps of elections).}
    \bigskip
  \begin{tabular}{l|c}
    \toprule
    Statistical Culture & Number of Elections \\
    \midrule
    Impartial Culture   &  $16$ \\
    Normalized Mallows  &  $48$ \\
    Urn                 &  $48$ \\
    \midrule
    Interval  &  $16$ \\
    Square    &  $16$ \\
    Cube      &   $16$ \\
    5D Cube   &  $8$ \\
    10D Cube  &  $8$ \\
    \midrule
    Circle    &  $16$ \\
    Sphere    &  $16$ \\
    \midrule
    SP (Con)  &  $16$ \\
    SP (Wal)  &  $16$ \\
    SPOC      &  $16$ \\
    \midrule
    GS (Bal)  &   $16$ \\
    GS (Cat)  &   $16$ \\
    \bottomrule
  \end{tabular}
\end{table}

\subsection{Synthetic Datasets}
We give the numbers of elections generated according to particular
statistical cultures and included in the basic dataset in
\Cref{tab:app:basic} (this dataset also includes special elections
$\ID$ and $\AN$, an approximate variant of $\UN$, several artificial
elections connecting these three, and the fourth special election,
$\ST$). The size-oriented and truncation-oriented datasets are readily
obtained from the basic one, because for each statistical culture the
number of elections generated according to it is divisible by $4$
(hence, the process of generating these datasets, described in
\Cref{sec:datasets}, is well-defined).  Our basic dataset is nearly
identical to the one used by
\citet{fal-kac-sor-szu-was:c:microscope}. The difference is that we
additionally include group-separable and SPOC elections, whereas they
also included several real-life elections from
Preflib~\citep{mat-wal:c:preflib}.

Regarding the comprehensive dataset, for most statistical cultures we
can also proceed as described in \Cref{sec:datasets} because the
number of elections generated according to them is divisible by
16. The only exception regards 5D and 10D Cube elections, of which we
have $8$. For each of these cultures, we have $2$ elections with
either $8$ or $16$ candidates and either $96$ or $192$ voters. Among
each pair of elections with the same size, one is left complete and
one is truncated either according to the top-$k$ method of according
to the random cut method (so, in particular, two of the 5D Cube
elections are top-$k$ truncated, two are random-cut truncated, and
four are left intact).

\section{Swap Failure}\label{app:swap-failure}
In this section, we provide examples of swap distance failures (see
also \Cref{app:proofs} for a proof of \Cref{prop:swap-triangle}).  We
consider a smaller variant of the size-oriented dataset, obtained just
like the original one, except that:
\begin{enumerate}
\item For each statistical cuture, we generate only half the number
  of elections as in the size-based dataset, each of them with 96
  voters.
\item Instead of considering elections with $8$ and $16$ candidates,
  we consider elections with $4$ and $8$ of them.
\end{enumerate}
We refer to it as the size-based/mini dataset.

In~\Cref{fig:map:swap-truncation} we present a map of the
size-based/mini dataset obtained using the truncation-based isomorphic
swap distance. Secondly, in~\Cref{fig:map:swap-del} we present a map
obtained using $\hat{d}^\del_\swap$, and
in~\Cref{fig:map:swap-del-triangle} we present the same map but
colored by the number of the triangle inequality violations (i.e, for
each point we calculated how many triangles including that point
violate the triangle inequality).

\section{Satisfaction of Consistency Axioms}\label{app:axioms}

In \Cref{tab:axiom:satisfaction},
we present the satisfaction of axioms defined in \Cref{sec:distances},
by two swap extensions, $\hat{d}^\del_\swap$ and $\hat{d}^\trun_\swap$ and positionwise distance, $\hat{d}_\pos$, defined in \Cref{sec:swap-positionwise} as well as DAP distance, $\hat{d}_\dap$, defined in \Cref{sec:feature}.
Note the impossibility result of \Cref{prop:swap-pos:impossibility}.

\begin{table}[t]
\def\arraystretch{1.5}
    \small
    \centering
    \caption{Axioms satisfied by the considered distance metrics.}
    \label{tab:axiom:satisfaction}
    \bigskip
    \begin{tabular}{c|cccc}
        \toprule
        metric & $\hat{d}^\del_\swap$ & $\hat{d}^\trun_\swap$ &
        $\hat{d}_\pos$ & $\hat{d}_\dap$ \\
        \midrule
        triangle inequality & \nmark & \ymark & \ymark & \ymark \\
        \midrule
        swap extension & \ymark & \ymark & \nmark & \nmark \\
        pos. extension & \nmark & \nmark & \ymark & \nmark \\
        \midrule
        $\ID$-consistency & \ymark & \nmark & \nmark & \ymark \\
        $\AN$-consistency & \ymark & \nmark & \nmark & \ymark \\
        $\UN$-consistency & \nmark & \nmark & \ymark & \nmark \\
        \bottomrule
    \end{tabular}
\end{table}

\section{Varying Election Sizes for DAP}\label{app:size-analysis}

\begin{figure}[b]
   \centering
   \includegraphics[width=0.8\linewidth]{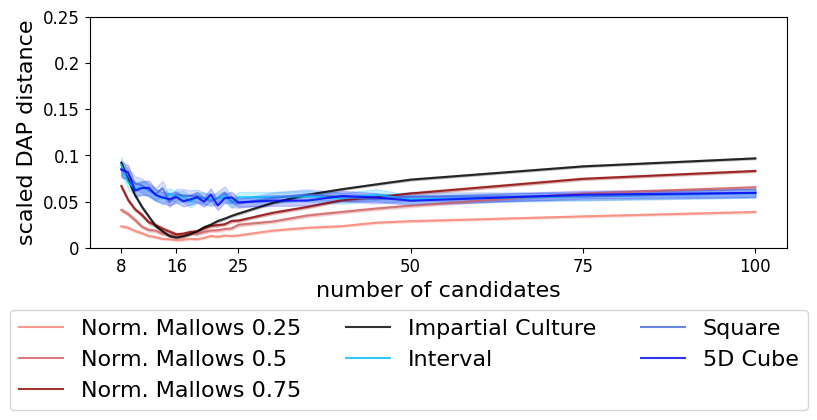}%
   \caption{Average DAP distance from size-16 elections to
     different-sized, complete elections from the same culture as a
     fraction of a max. distance in our dataset.}
   \label{fig:rob:dap-100}
\end{figure}

For the case of DAP, we have extended the first experiment from
\Cref{sec:quantitative-exp} to regard elections with up to $100$
candidates (specifically, we have also included numbers of candidates
in the set $\{25,30,35,40,45,50,75, 100\}$). We show the resulting
plot in \Cref{fig:rob:dap-100}. We see that as the number of
candidates increases, the plots flatten out; e.g., for Euclidean
elections at a bit over $6\%$ and for IC at a bit over $10\%$.

\section{Truncation Types Analysis}
\label{app:truncation:analysis}

In this section, we give more in-depth analysis of the effect of
different truncation types on positionwise and DAP distances.

\begin{figure*}[t]
    \centering
    \includegraphics[width=0.95\linewidth]{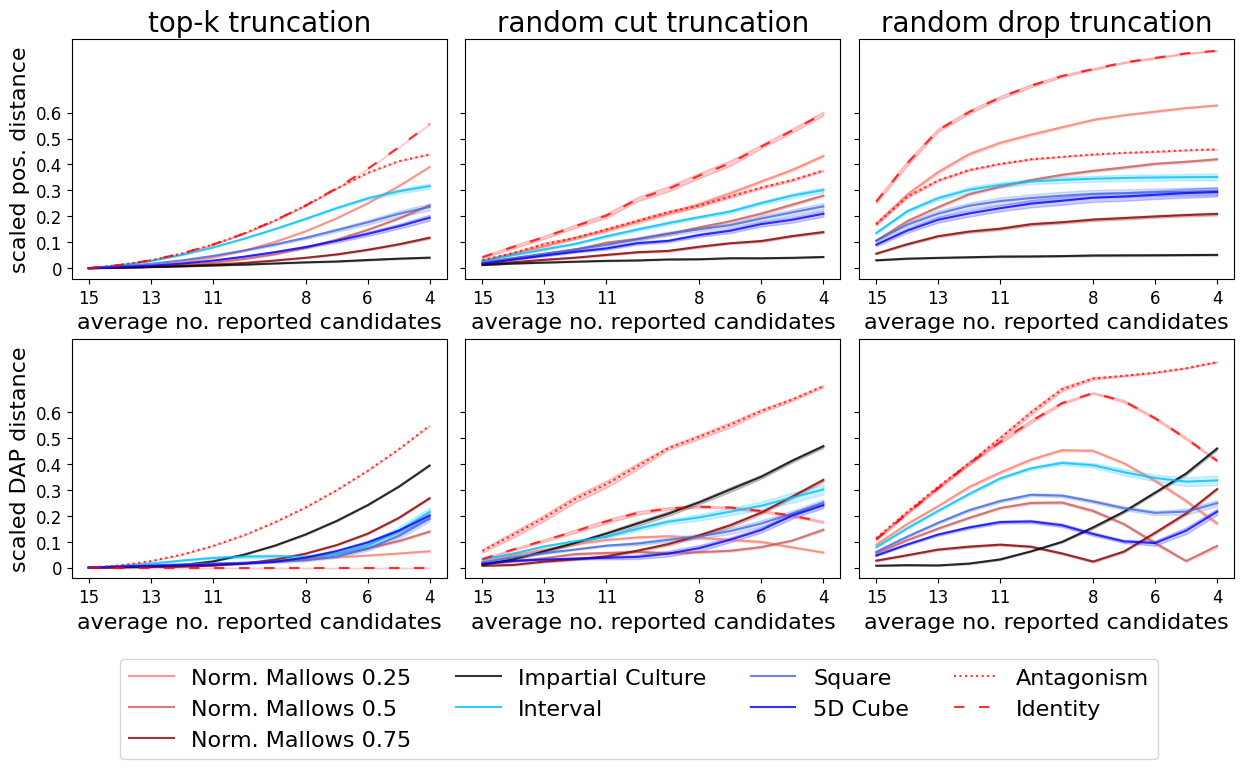}
    \caption{Average positionwise (top row) or DAP (bottom row) distances from a complete election to its top-$k$ (left column), random cut (center column), or random drop (right column) truncations as a fraction of a maximal distance in our dataset. Each datapoint is an average of 25 samples, the shaded area corresponds to the 95\% confidence interval.}
    \label{fig:trunc_dist:all}
\end{figure*}

\Cref{fig:trunc_dist:all} presents average positionwise and DAP
distances between a complete election and its different top-truncated
variants.  The plots were obtained in the second experiment from
\Cref{sec:quantitative-exp}.

For positionwise distance (top row of \Cref{fig:trunc_dist:all}), the
most affected by truncations are the cultures with the strongest
structure (i.e., Mallows with a small parameter $\normphi$ and
interval elections) and the effect decreases when the elections become
more chaotic.  In particular, impartial culture is almost not affected
at all by any of our truncation methods.  The reason lays in the way
positionwise distance handles the truncated votes: in frequency
matrices, the values for unreported candidates are spread across
several last positions.  Hence, a frequency matrix of a heavily
truncated election, becomes similar to the frequency matrix of a
uniformity election.

Top-$k$ truncations (top left picture) and random cut truncations (top
center picture) seem to have similar effect on positionwise distance.
In both, the average distance is gradually increasing for all of the
cultures (except impartial culture) as we truncate the election more
and more.

On the other hand, for random drop truncations (top right picture),
the distances quickly become large for the most structured cultures.
In particular, when the votes have at most 8 reported candidates on
average, for Mallows elections with $\normphi=0.25$, the average
distance is larger than the half of the diameter.  This suggests that
random drop truncation brings much more chaos into elections than the
other two truncation methods.  Indeed, under top-$k$ and random cut
truncations the candidates in a few top positions of a vote will
rarely end up in the truncated part.  However, under random drop
truncation there is a significant probability that we move a top
candidate to a truncated part and another candidate becomes a top one.

For the DAP distance (bottom row of \Cref{fig:trunc_dist:all}), the
question of which cultures are the most affected by truncation is no
longer simple and consistent across the truncation types.  For top-$k$
truncation (bottom left picture), for all cultures, the average
distance is gradually increasing as we increase the range of
truncation.  Interestingly, the cultures most affected by top-$k$
truncation under DAP distance are exactly those that were least
affected under positionwise distance (i.e., impartial culture and
Mallows with $\normphi=0.75$).  This is due to different treatment of
truncations by DAP and positionwise distance (see
\Cref{sec:pos-dap:diff}).  In heavily truncated elections voters agree
that many pairs of candidates are equal, hence such elections are
considered to be similar to identity.  Thus, the most chaotic cultures
are far away from such elections under DAP distance.

For random drop truncation (bottom right picture), we see a very
interesting pattern for several cultures: the distance quickly grows
as the average number of reported candidates increases, but then it
drops (and possibly increases again at the end).  As discussed before,
random drop truncation brings a lot of chaos into elections as
candidates can move downward (when they land in the truncated part) or
upward (when we drop candidates before them) in the votes.  This
explains a quick increase in the distance for highly structured
cultures, such as Mallows with $\normphi=0.25$ or interval elections,
where random drop truncation can destroy the structure.  However, as
votes become heavily truncated they start to become similar to each
other in the sense that a lot of candidates are jointly seen as
equally bad.  Thus, such elections become similar to identity, which
explains the movements of cultures in the right hand side of the plot.

Finally, for random cut truncation (bottom center picture) the plots
look rather similar to that for top-$k$ truncation.  The main
differences are a bit larger average distances altogether and that the
plot for Mallows with $\normphi=0.25$ resembles a bit its plot for
random drop truncation but with a smaller amplitude.  The similarity
can be explained by the fact that random cut truncation is in some
sense similar to top-$k$ truncation (the top part of each vote is
exactly the same as in the original, complete election).  The
differences come from the fact that random cut truncation brings a bit
more chaos than top-$k$ truncation since the votes can be cut in
different points.

\section{Extended Maps of Elections}\label{app:maps-of-elections}
Here we extend \Cref{sec:maps} and show maps of elections for the size-oriented, comprehensive and random drop datasets under both positionwise and DAP distances (see \Cref{fig:big-map-experiments-in-appendix}). The random drop dataset was obtained by taking the basic one
and applying random drop truncation to half of the elections from each
statistical culture (so that, in expectation, each voter ranks half of
the candidates). 

For the size-oriented and comprehensive datasets, both DAP and positionwise maps look similar
to the maps in \Cref{fig:map:trun1-dap-ape} and \Cref{fig:map:trun1-poswise}, respectively. Their shapes are preserved and truncated elections again cluster around UN and impartial culture under positionwise and around high dimensional Euclideans under DAP. Maps drawn from the random drop dataset show a similar overarching shape, but there the effect of truncation is much more prominent. Under both
distances, the truncated elections form a dense cluster. This means that
all such elections are much more similar in structure to each other
than to complete elections from their respective cultures.  This agrees
with our observations from \Cref{fig:trunc_dist:all}, where for the
average of half of the candidates reported (which is exactly the case
for the random drop truncated elections in our dataset) for many
cultures a complete election is on average far away from its random
drop truncated version (e.g., for Mallows with $\normphi=0.25$, it is
around half of the diameter for both distances).

\begin{figure*}[t]
  \centering
  \begin{subfigure}[t]{0.32\textwidth}
    \centering
    \includegraphics[width=\textwidth, trim={0.2cm 0.2cm 0.2cm 0.2cm}, clip]{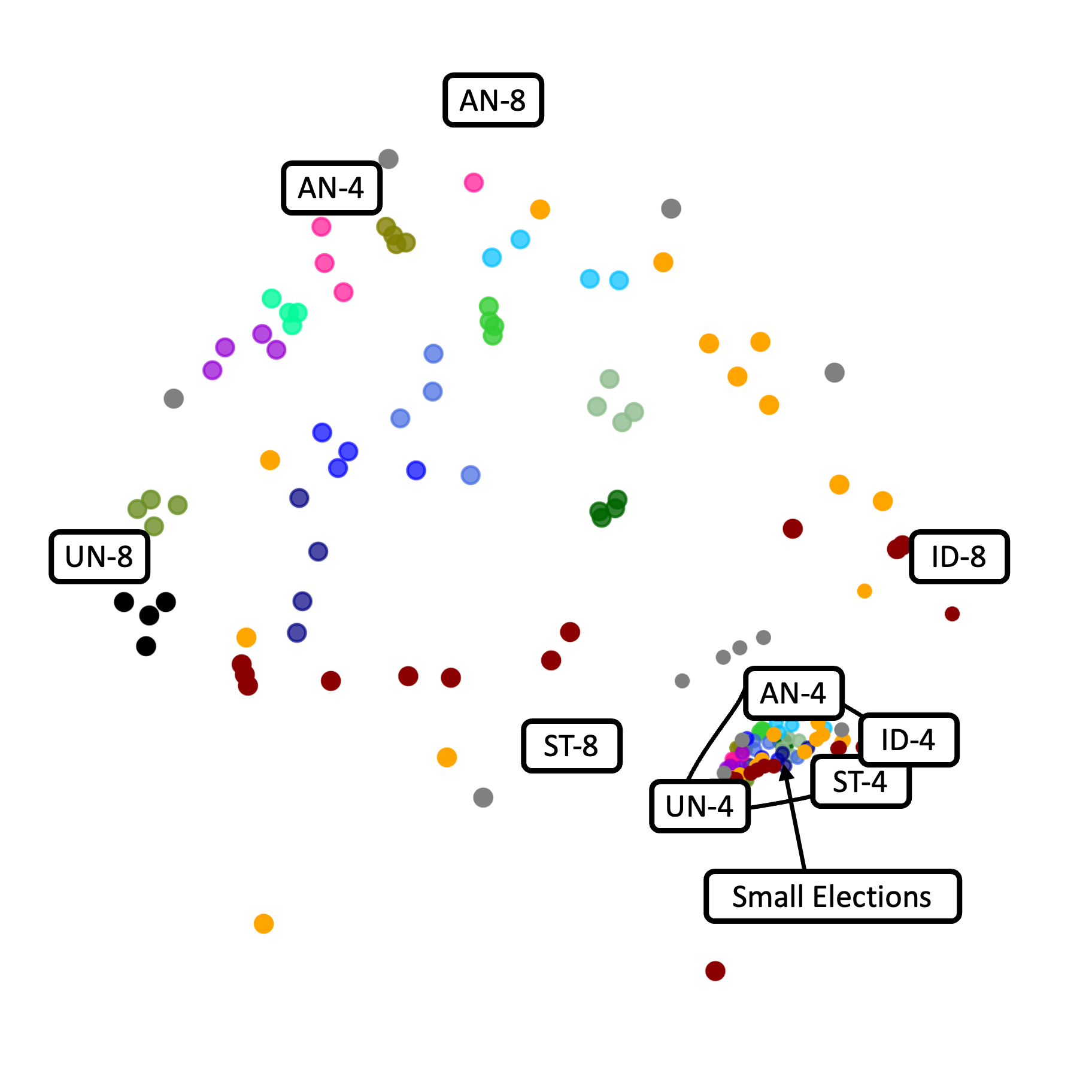}%
    \caption{Map obtained using the truncation-based isomorphic swap distance.}
    \label{fig:map:swap-truncation}
  \end{subfigure}
    \begin{subfigure}[t]{0.32\textwidth}
    \centering
    \includegraphics[width=\textwidth, trim={0.2cm 0.2cm 0.2cm 0.2cm}, clip]{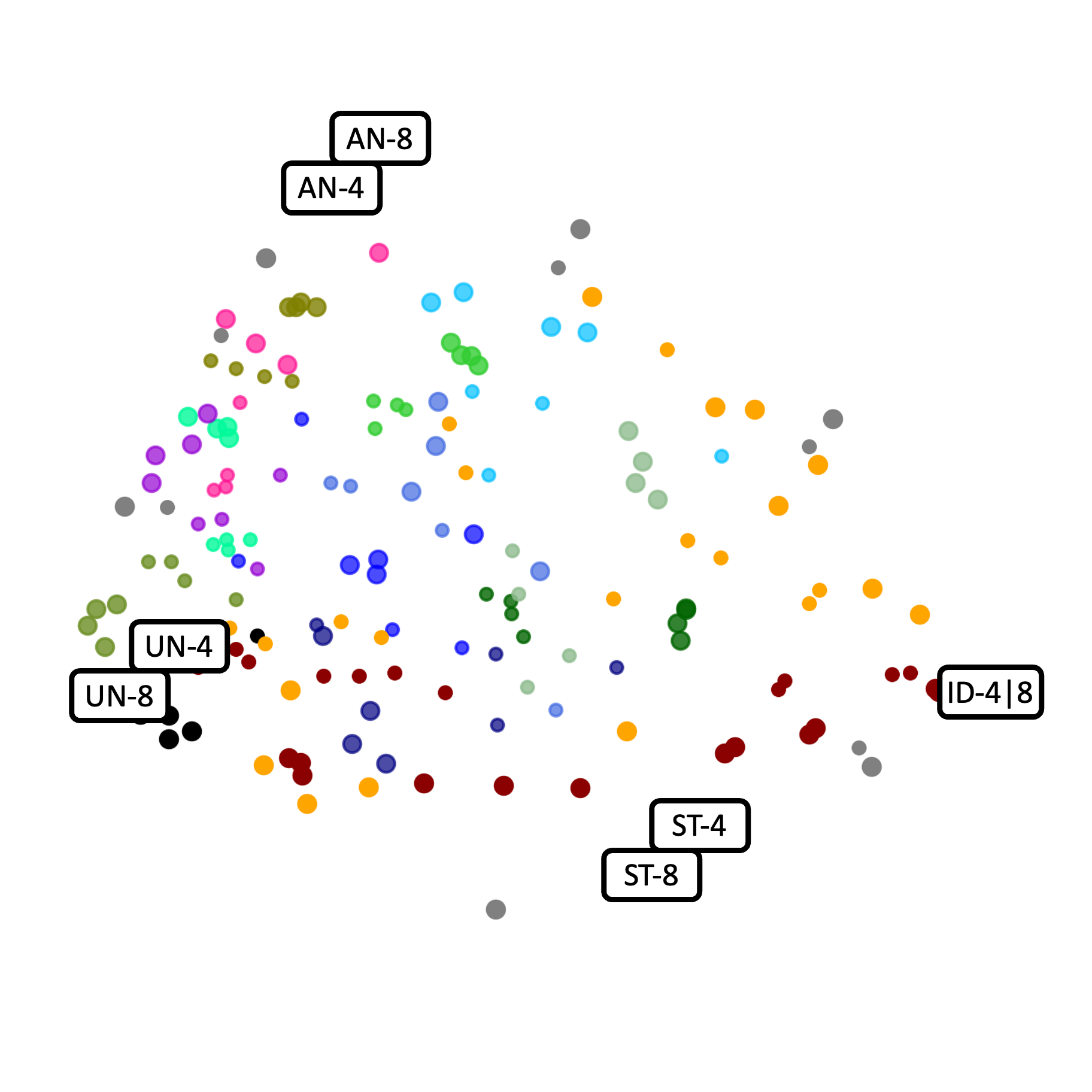}%
    \caption{Map obtained using $\hat{d}^\del_\swap$.}
    \label{fig:map:swap-del}
  \end{subfigure}
  \begin{subfigure}[t]{0.35\textwidth}
    \centering
    \includegraphics[width=\textwidth, trim={0.2cm 0.2cm 0.2cm 0.2cm}, clip]{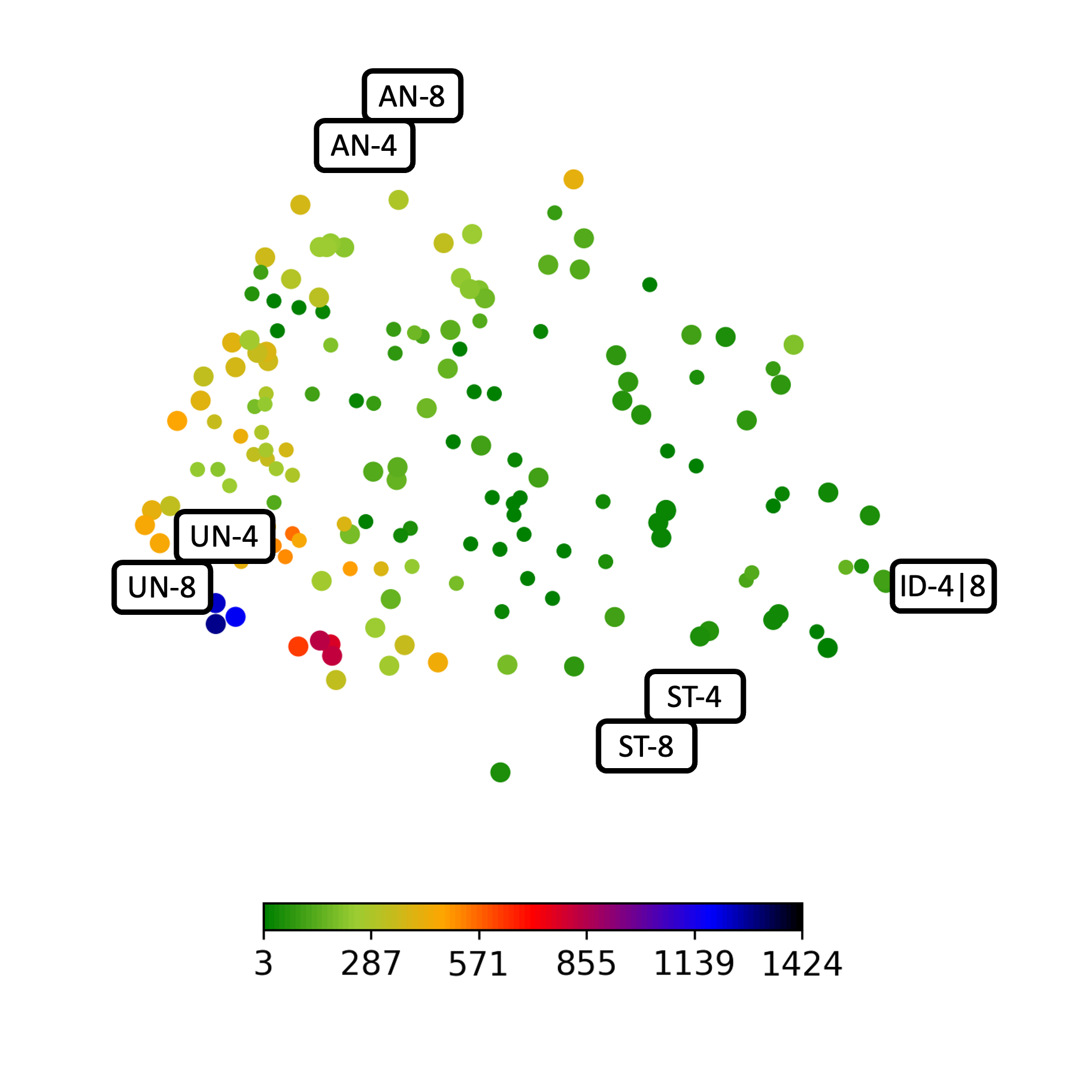}%
    \caption{Map obtained using $\hat{d}^\del_\swap$, colored by no. triangle inequality
      violations
      (for each point we calculated how many triangles including that point violate triangle inequality).}
    \label{fig:map:swap-del-triangle}
  \end{subfigure}

  \caption{\label{fig:swap-issues}Maps showing various failures of the isomorphic swap distance extensions on the size-based/mini dataset.}
\end{figure*}

\begin{figure*}[t]
  \centering
  \begin{subfigure}{0.32\linewidth}
    \includegraphics[width=\linewidth]{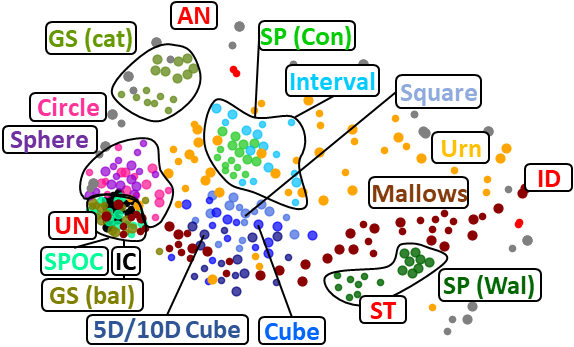}%
    \caption{Size-oriented (positionwise).}
    \label{fig:map:sizes-poswise}
  \end{subfigure}
  \begin{subfigure}{0.32\linewidth}
    \includegraphics[width=\linewidth]{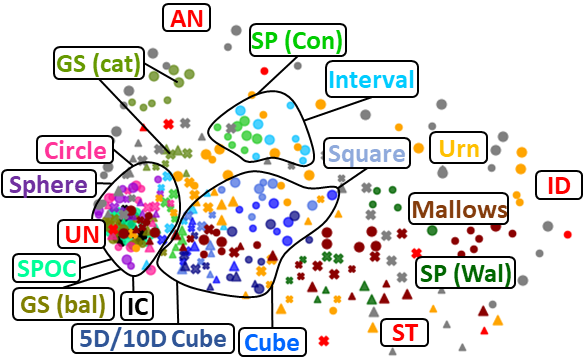}%
    \caption{Comprehensive (positionwise).}
    \label{fig:map:full-poswise}
  \end{subfigure}
   \begin{subfigure}{0.35\linewidth}
     \includegraphics[width=\linewidth]{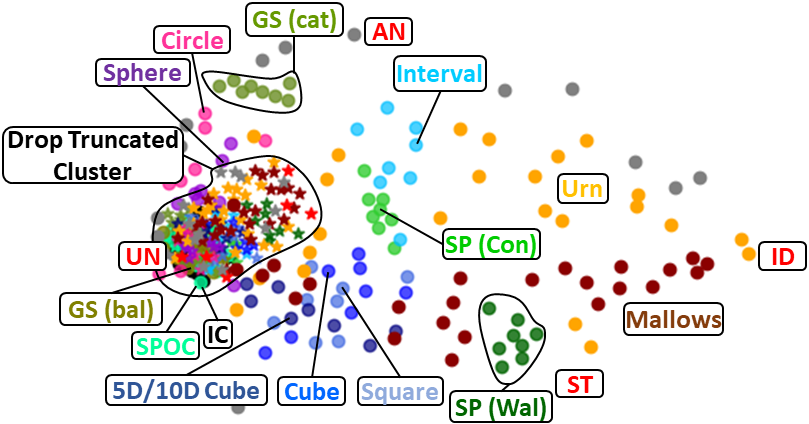}
     \caption{Random drop (positionwise).}
     \label{fig:map:trunc_drop:posiwse}
   \end{subfigure}
  %
  %
  \begin{subfigure}{0.32\linewidth}
    \includegraphics[width=\linewidth, trim={0cm 0cm 0cm 0cm}, clip]{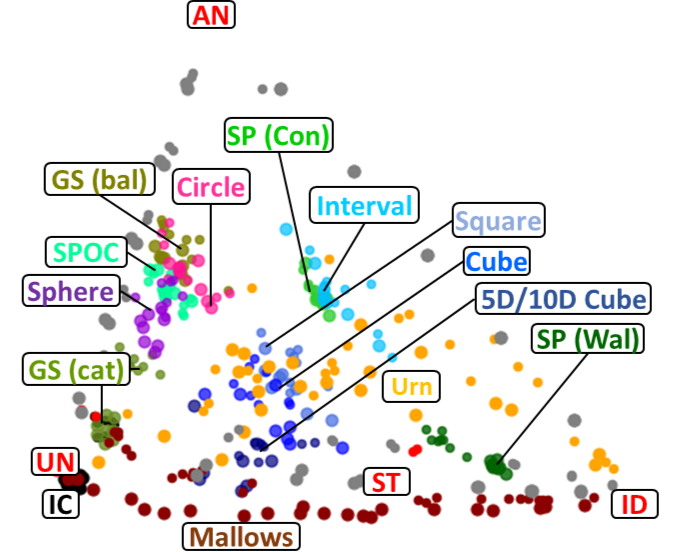}%
    \caption{Size-oriented (DAP).}
    \label{fig:map:sizes-dap-ape}
  \end{subfigure}
  \begin{subfigure}{0.32\linewidth}
    \includegraphics[width=\linewidth, trim={0cm 0cm 0cm 0cm}, clip]{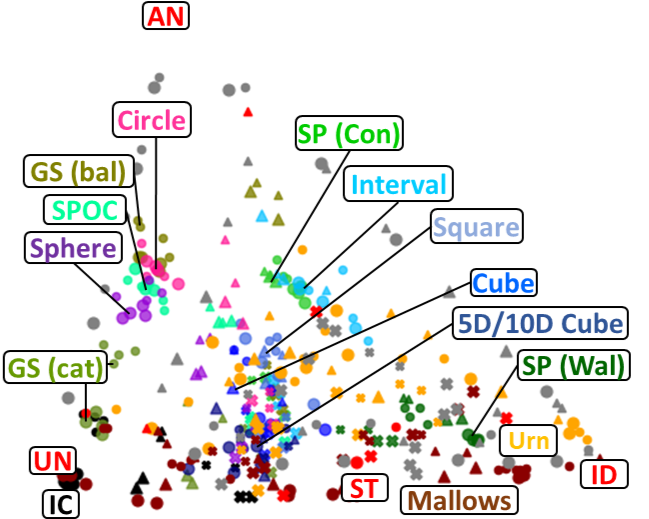}%
    \caption{Comprehensive (DAP).}
    \label{fig:map:full-dap-ape}
  \end{subfigure}
   \begin{subfigure}{0.35\linewidth}
     \includegraphics[width=\linewidth]{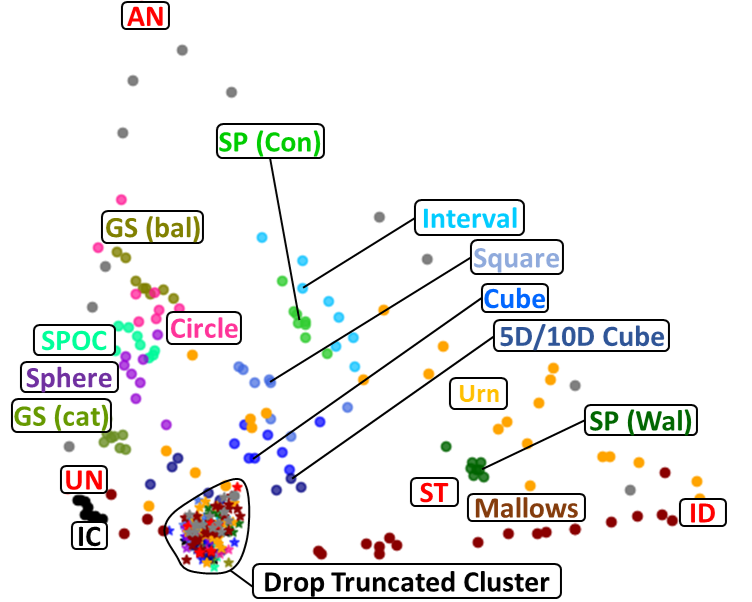}%
     \caption{Random drop (DAP).}
     \label{fig:map:trunc_drop:dap}
   \end{subfigure}
  
  \caption{Maps of elections created using the positionwise (top) and
    DAP (bottom) distances, for the size-oriented (left),
    comprehensive (center), and  random drop
    (right) datasets.  Top-$k$ truncated elections are marked with
    triangles, random-cut truncated ones with crosses, random drop ones with stars and complete
    ones with circles.  Larger markers mean elections with 16
    candidates, and smaller markers mean elections with 8 candidates.
    In the ST 
    election each voter ranks the same half of the candidates on top,
    but otherwise the votes are chosen uniformly~at~random.}
  \label{fig:big-map-experiments-in-appendix}
\end{figure*}

\section{Extended Discussion of the Map of Preflib}
\label{app:preflib:all}

In this section we provide a number of missing details and additional
analyses regarding the map of Preflib---as well as its several
colorings---and the observations we make.  First, we describe the
contents of these map(s).

\begin{figure}[t!]
  \centering
  \includegraphics[width=0.7\columnwidth]{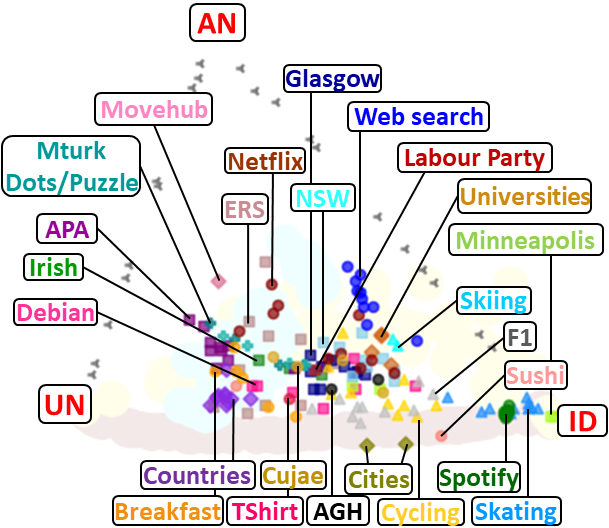}%
  \caption{\label{fig:preflib_colors_dap}Map of Preflib, where colors
    and markers indicate the specific Preflib datasets from which each
    given election comes from. Large pale discs show synthetic
    elections generated using Mallows (pale brown), urn (pale yellow)
    or Euclidean elections with points generated uniformly at random
    from hypecubes (pale blue).}
\end{figure}

\subsection{Preflib Dataset Composition}\label{app:preflib}

We show several maps of Preflib. One in \Cref{fig:map:preflib}, where
all the Preflib election are represented as black dots for clarity,
one in \Cref{fig:preflib_colors_dap}, where for each Preflib election
we indicate which exact Preflib dataset it comes from, several of them
in \Cref{fig:preflib-types} colored by election type, and one in
\Cref{fig:preflib_fragment}, where we present a smaller fragment of Preflib
(which we consider for the sake of the positionwise distance, see
\Cref{app:poswise:preflib}). Next we describe the details of the
Preflib dataset.

To construct the set of Preflib elections to include in our maps, we
considered every dataset on Preflib with an \textit{Election} tag and
\textit{.soc} or \textit{.soi} file format (Preflib includes also other
kinds of preference data, like matching data, and other preference
formats, like cardinal or categorical preferences, that were not within
the scope of this paper).  We note that datasets vary greatly in the
number of datafiles they contain (e.g., there is one file in the
``T-Shirt'' dataset, but 362 files in the ``Spotify Daily Chart''
dataset).  To avoid overrepresentation of elections from datasets with
numerous files in the map, from every dataset with more than 10 files,
we have randomly chosen 10 files to include in the map.

We have also omitted some of the datasets introduced to Preflib by
\citetapp{boe-sch:c:real-life-elections}.  The nature of many of these
elections is very specific, as they report the given ranking evolving
over time (e.g., in the ``Spotify Countries Chart'' dataset each file
corresponds to a given country and a month and each vote represents
the ranking of the 200 most listened songs on Spotify in this country
on a particular day of this month, or in the ``Formula 1 Races''
dataset each file corresponds to a given Formula 1 race and each vote
represents the ranking of drivers by the time in which they finished a
particular lap).

Finally, to show which Preflib elections correspond to which
statistical cultures, we added all cultures from the standard dataset
and generated elections from these with 20 candidates and 1000 voters
(number of candidates was higher here than in the standard dataset to
adapt to Preflib elections, which on average tend to have a higher
number of candidates). On our maps, we show these elections as much
larger, pale discs, so that the Preflib elections are visible on top
of them (depending on the particular maps, we either used the
same---but more pale---colors as in the maps of synthetic elections,
or we collapsed some of the colors into one; e.g., using the same pale
blue for all Euclidean elections with points generated uniformly at
random from hypercubes of various dimensions).  In some of the maps,
such as those in \Cref{fig:preflib_colors_dap} and
\Cref{fig:preflib-types}, we only kept synthetic elections generated
using Mallows, urn, and Euclidean models (with points generated from
hypercubes).  For orientation purposes, we have added $\ID$, $\AN$,
and an approximation of $\UN$ elections as well as elections forming
paths between these three, each with 16 candidates and 500 voters.

\Cref{tab:preflib_elections} at the end of the appendix presents the
exact composition of our Preflib dataset (excluding the artificial
elections added for orientation and the synthetic ones). For each
election, we report the corresponding Preflib file, the marker we use
to denote its position on the maps in
\Cref{fig:preflib_colors_dap,fig:preflib_fragment,fig:preflib-types},
the number of candidates, the number of votes, the number of unique
votes, and the values of Diversity, Agreement, and Polarization
indices, as well as the name of the statistical culture (and its
parameters) of the synthetic election closest to it (in terms of DAP)
in the dataset.

\subsection{Types of Elections}\label{app:preflib-types}

\begin{figure}[t!]
  \centering
  \begin{subfigure}[b]{0.4\columnwidth}
  \centering
    \includegraphics[width=\columnwidth]{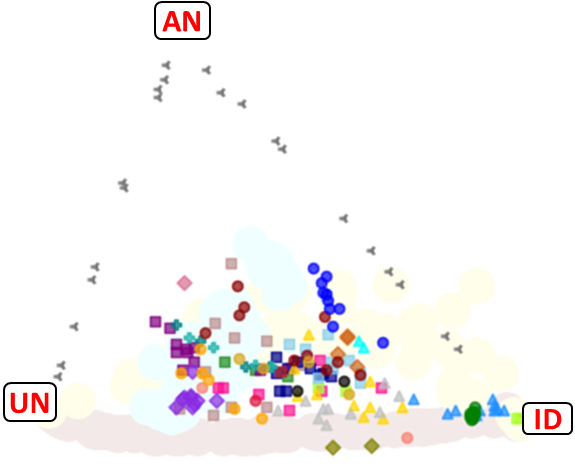}%
    \caption{All types of elections.}
    \label{fig:preflib:alltypes}
  \end{subfigure}~\quad
  \begin{subfigure}[b]{0.4\columnwidth}
  \centering
    \includegraphics[width=\columnwidth]{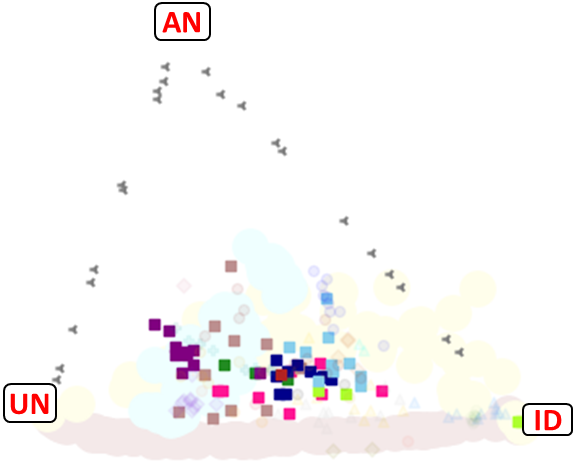}%
    \caption{Political (squares).}
    \label{fig:preflib:political}
  \end{subfigure}
  \bigskip

  \centering
  \begin{subfigure}[b]{0.4\columnwidth}
  \centering
    \includegraphics[width=\columnwidth]{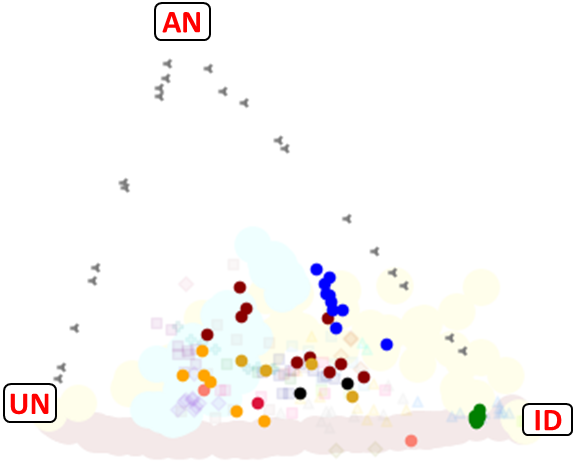}%
    \caption{Surveys (dots).}
    \label{fig:preflib:surveys}
  \end{subfigure}~\quad
  \begin{subfigure}[b]{0.4\columnwidth}
  \centering
    \includegraphics[width=\columnwidth]{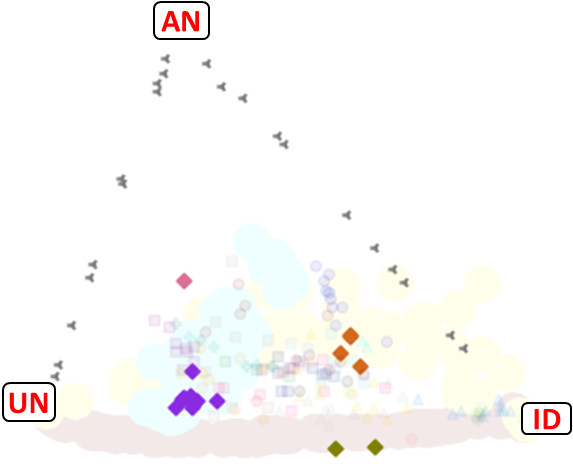}%
    \caption{Indicator (diamonds).}
    \label{fig:preflib:indicator}
  \end{subfigure}
  \bigskip

  \centering
  \begin{subfigure}[b]{0.4\columnwidth}
  \centering
    \includegraphics[width=\columnwidth]{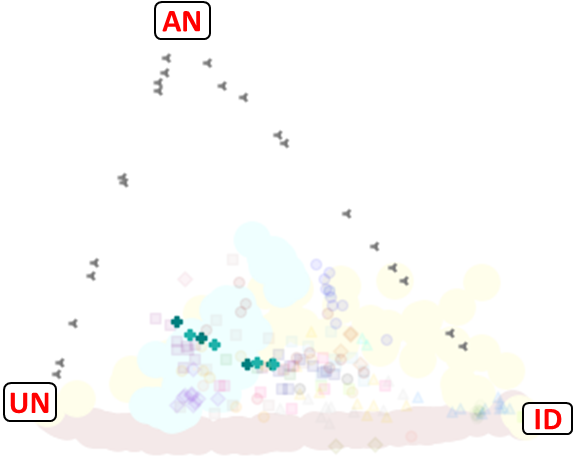}%
    \caption{Mechanical Turk (crosses).}
    \label{fig:preflib:mturk}
  \end{subfigure}~\quad
  \begin{subfigure}[b]{0.4\columnwidth}
  \centering
    \includegraphics[width=\columnwidth]{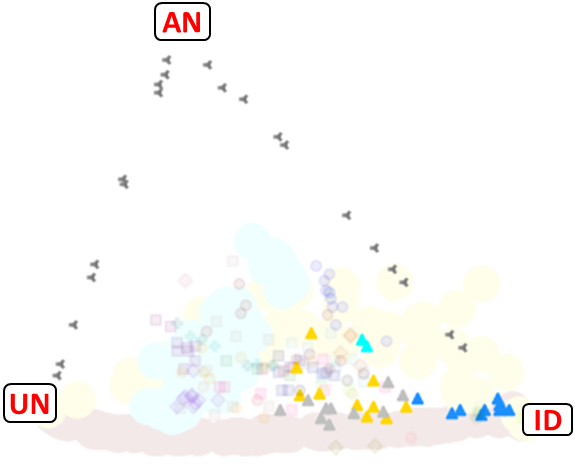}%
    \caption{Sports (triangles).}
    \label{fig:preflib:sports}
  \end{subfigure}
  \bigskip
  
  \caption{Map of Preflib with highlighted different types of elections.
    Large pale discs show synthetic
    elections generated using Mallows (pale brown), urn (pale yellow)
    or Euclidean elections with points generated uniformly at random
    from hypecubes (pale blue).
}
\label{fig:preflib-types}
\end{figure}

Following the information provided on Preflib, as well as the approach
of \citeauthor{boe-bre-fal-nie-szu:c:compass}~\citep{boe-bre-fal-nie-szu:c:compass,boe-bre-elk-fal-szu:c:frequency-matrices}, we have classified
the Preflib elections into five types (elections of each type are presented on the maps using their specific markers):
\begin{description}
\item[Political elections (marked with squares).] By political elections, we mean elections
  where people cast votes over other people, to select someone---or a
  group of individuals---for some sort of a leadership position. These
  elections include, e.g., elections for the city council of Glasgow,
  elections in several districts in Ireland for their general
  elections, but also a number of elections from professional
  organizations, including those collected by Electoral Reform
  Society, elections for the president of the American Psychological
  Association, APA, or leadership elections held by the Debian
  project.

\item[Surveys (marked with dots).] Survey elections represent situations where people
  were asked to express their preferences over a number of items. This
  includes, e.g., the classic Sushi dataset~\citep{kam:c:sushi}, where
  people ranked sushi types, or surveys over breakfast items, opinions
  on movies, students' preferences over university courses (AGH and
  Cujae) or various aspects of education (Cujae).

\item[Indicator (marked with diamonds).] Indicator elections rank various objects based on
  more-or-less objective criteria. This includes, e.g., rankings of
  universities according to different benchmarks, rankings of
  countries, or cities according to various quality-of-life measures.

\item[Mechanical Turk (marked with crosses).] These elections include rankings collected
  using Mechanical Turk by Andrew Mao. The tasks involved solving
  simple puzzles.

\item[Sports (marked with triangles).] Sports elections represent results of various
  competitions. For example, in the Formula~1 dataset each election
  represents a single season and each vote represents a single
  race. Sports elections also represent data on cycling races,
  skiing, and skating.
\end{description}

In \Cref{fig:preflib-types} we show maps of Preflib where only
elections of a given type are highlighted, whereas all the other
ones---including those generated synthetically and represented as
large discs---are very pale. These maps lead to some interesting
conclusions. Foremost, political (\Cref{fig:preflib:political}) and
survey elections (\Cref{fig:preflib:surveys}) seem to be occupying
roughly the same area of the map. One could point out to some minor
differences (e.g., the political APA elections take an area that is
not well-represented by the surveys elections, and the surveys
WebSearch dataset takes an area that political elections do not often
take) but, overall, these do not seem very relevant. Thus in
experiments the two types of elections might be treated similarly
(however, one should explore reasons for their similarity on our map
in more detail). The indicator elections
(\Cref{fig:preflib:indicator}) take very specific, but quite spread
out positions on the map. Given how few datasets in this category we
have, it is not clear if we can draw any substantial
conclusions. Rather, it seems that this type of election may not stand
out so much from the previous two types (but it might also be possible
that indicator elections tend to assume very specific positions and we
would see this more strongly given more data). Mechanical Turk
elections (\Cref{fig:preflib:mturk}) also assume very specific positions, but also this might be
simply because we have very few elections of this type, collected in
very specific way (see the Preflib website for a description).

The only type of elections that truly stands out are the sports ones
(\Cref{fig:preflib:sports}).  Essentially all of them are closer to
$\ID$ than to $\UN$ (and, indeed, sometimes very close to $\ID$, as in
the case of the Skating dataset). This is natural as, indeed, we
expect the sportsmen to perform similarly in different
competitions. However, of course, the figure skating competitions held
over a single day have much stronger correlation (position on the map
very close to $\ID$) than Formula~1 seasons, held over the course of a
year, where drivers' abilities and performances can vary from a race
to a race (similarly for cycling, where some participants are better
on the flat areas of the race, and some do better in the climbing
parts).

\subsection{Parameters of Mallows and Urn Elections Close to Preflib}\label{app:preflib-ranges}

\begin{figure*}[t!]
  \centering
  \includegraphics[width=\linewidth]{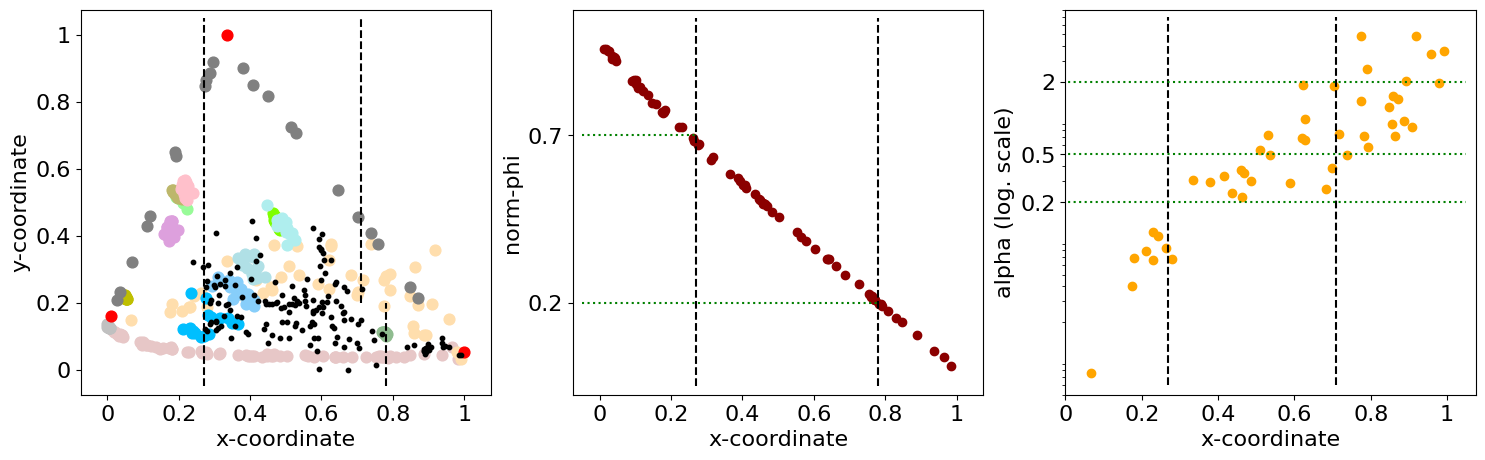}
  \begin{subfigure}[t]{0.32\linewidth}
  \centering
  \caption{Map of Preflib elections. We attempt to capture the
    elections between dashed lines using the Mallows and urn
    models. Black dots represent Preflib elections, colorful ones
    represent synthetic ones (pale brown is for Mallows, pale
    yellow/orange for urn; other colors are as in~\Cref{fig:map:preflib}).}
    \label{fig:preflib:ranges-mallows-urn}
  \end{subfigure}\quad
  \begin{subfigure}[t]{0.31\linewidth}
  \centering
    \caption{Mallows $\normphi$ range. For each Mallows election
      (i.e., each dot), its $x$ coordinate is the $x$-position on the
      map, and its $y$-coordinate is its $\normphi$ parameter.}
    \label{fig:preflib:ranges-mallows}
  \end{subfigure}\quad
  \begin{subfigure}[t]{0.3\linewidth}
  \centering
    \caption{Urn parameter range. For each Urn election (i.e., each
      dot), its $x$ coordinate is the $x$-position on the map, and its
      $y$-coordinate is its contagion parameter.}
    \label{fig:preflib:ranges-urn}
  \end{subfigure}
  \caption{\label{fig:preflib:ranges}Ranges of $\normphi$ for the Mallows model and contagion parameters for the urn model that seem to
    capture Preflib elections well.  }
\end{figure*}

We have attempted to establish the ranges of the $\normphi$ and
contagion parameters for the Mallows and urn models that best capture
many of the Preflib elections (omiting those very close to $\ID$.  In
\Cref{fig:preflib:ranges-mallows-urn} we show a map of Preflib where
we indicate which elections we want to capture.  Specifically, we seek
$\normphi$ parameters of the Mallows elections whose $x$-positions on
the map are between those shown by the left dashed line and the
bottom-right one. Similarly, we seek contagion parameters of those urn
elections whose $x$-coordinate is between the left and the top-right
dashed lines. In
\Cref{fig:preflib:ranges-mallows,fig:preflib:ranges-urn} we show
relation between the $x$-positions on the map and the $\normphi$ and
contagion parameters of the Mallows and urn elections, respectively.
From these plots, we read off that the $\normphi$ parameters that we
are after are in the range $[0.2,0.7]$. For the case of the urn model,
we can consider several options. If we choose the range $[0.2,0.5]$ of
contagion parameters, then most of the generated elections will fall
into the desired range of $x$-positions on the map, but we will miss
some other urn elections in the range. If we consider range $[0.2,2]$,
then we get all urn elections from the desired area on the map, but
also a number of others.

As a further validation of our recommendations, for each of the Preflib
elections we found the closest synthetic election in the dataset and
recorded the statistical culture from which it comes from, as well as the
parameter value for which it was generated (see the last column in \Cref{tab:preflib_elections}). In most cases, these were either Mallows
or urn elections, with parameters from the recommended ranges (but there were also
cases where, e.g., an urn election with the contagion parameter around 3.62 was closest to a Preflib election\footnote{This was election \texttt{00018-00000003.soi} from the Minneapolis dataset. However, this election is very particular as it captures a political elections with allowed write-ins. In effect, the election has nearly 500 candidates, including Mickey Mouse, that are mostly unranked.}). However,
hypercube elections (mostly for 5 and 10 dimensions) were also prominently
represented, which justifies their use in experiments.

\begin{figure*}[t!]
  \centering
  \begin{subfigure}[b]{\columnwidth}
  \centering
    \includegraphics[width=0.75\columnwidth]{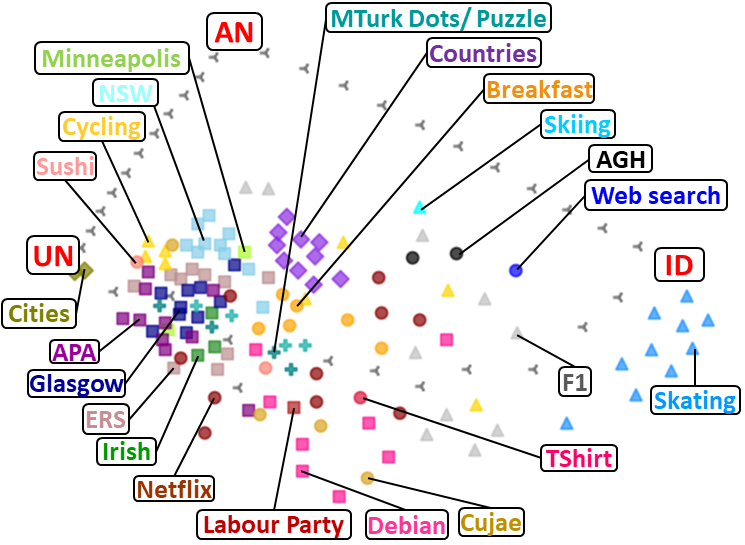}%
    \caption{Positionwise.}
    \label{fig:preflib_fragment_positionwise}
  \end{subfigure}
  \begin{subfigure}[b]{\columnwidth}
  \centering
    \includegraphics[width=0.75\columnwidth]{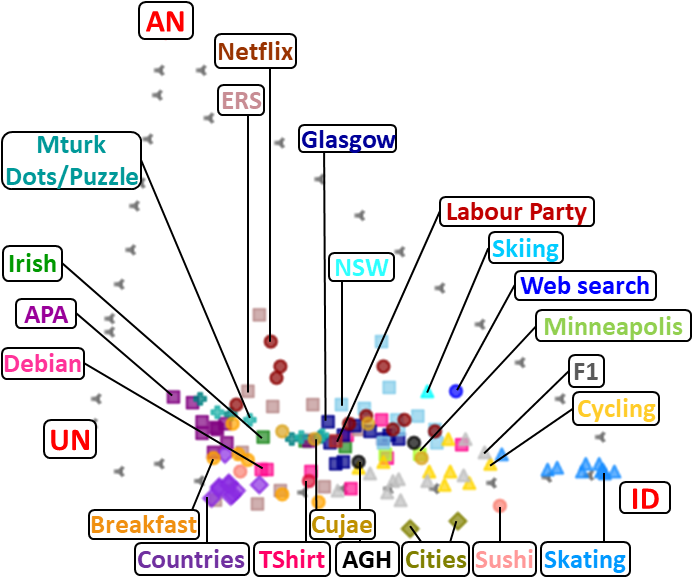}%
    \caption{DAP.}
    \label{fig:preflib_fragment_dap}
  \end{subfigure}
  \caption{\label{fig:preflib_fragment}Maps of Preflib fragment.
}
\end{figure*}

\subsection{Positionwise Map of Preflib}
\label{app:poswise:preflib}

Obtaining the map of Preflib elections similar to the one in
\Cref{fig:map:preflib} but using positionwise distance is
computationally difficult.  This is due to large numbers of candidates
in some of the elections, which leads to large sizes of frequency
matrices. In this appendix, we present a positionwise map of a subset
of elections from our Preflib dataset.  Specifically, we removed all
elections with more than 200 candidates (see \Cref{tab:preflib_elections} for an
exact composition of our Preflib dataset).
\Cref{fig:preflib_fragment_positionwise} presents a positionwise
distance map for this fragment of the dataset.  For comparison, we
have also included DAP distance map of the same fragment in
\Cref{fig:preflib_fragment_dap}.

We see that in positionwise map many elections are much closer to $\UN$ than in the DAP map.
This can be explained by a different treatment of truncated elections by both distances (see \Cref{sec:pos-dap:diff}).
In particular, this affected many political elections, like elections from New South Wales (NSW) or Glasgow elections,
which very often are heavily top-truncated
(as voters report only a few top candidates).

\appendixproof{Proposition}{prop:emd-vs-w-in-text}{%
  This proof is available in \Cref{sec:emd-vs-W}.}

\appendixproof{Proposition}{prop:swap-pos:impossibility}{
\begin{proof}
    First, let us show that there is no swap extension that satisfies $\ID$-, $\AN$- and $\UN$-consistency.
    Assume, by contradiction, that there is such a swap extension $\hat{d}^*$.
    We know that
    $d_\swap(\AN_{3,6},\UN_{3,6}) = 4/9 \neq 11/18 =
    d_\swap(\AN_{4,24},\UN_{4,24})$.  Consequently,
    $\hat{d}^*(\AN_{3,6},\UN_{3,6}) \neq
    \hat{d}^*(\AN_{4,24},\UN_{4,24})$ as well.  Since, by $\AN$- and
    $\UN$-consistency,
    $\hat{d}^*(\AN_{3,6},\AN_{4,24})=\hat{d}^*(\UN_{3,6},\UN_{4,24})=0$,
    we get a violation of the triangle inequality.

    Next, let us focus on positionwise extensions.
    To this end, let us establish the distances between $\UN$ and $\ID$ as well as $\UN$ and $\AN$ for equal and even number of candidates $m$.
    Denote uniformity vector
    $\vec{u} = (\nicefrac{1}{m},\nicefrac{1}{m},\dots,\nicefrac{1}{m})$,
    identity vectors $\vec{e}_i \in \mathbb{R}^m$ with $1$ at $i$th position and zero otherwise,
    and antagonism vectors $\vec{a}_i \in \mathbb{R}^m$
    with $\nicefrac{1}{2}$ at $i$th and $m-i$th positions and zero otherwise.
    It is easy to verify that for $m=2k$, it holds that
    \begin{align*}
        W(\vec{u},\vec{e}_i)&=\frac{1}{2m^2} \big( (i-1)^2 + (m-i)^2 + m-1 \big) \quad
        \mbox{and}\\
         W(\vec{u},\vec{a}_i)&=\frac{1}{4k^2} \left( (i-1)^2 + (k-i)^2 + \textstyle\frac{(i-1)^2 + (k-i)^2}{k-1}\right).
    \end{align*}
    By summing the above equations for all $i \in [m]$, we obtain that
    \begin{align*}
        {d}_\pos(\UN_m,\ID_m)&=\frac{1}{3}-\frac{1}{3m^2} \quad
        \mbox{and}\\
        {d}_\pos(\UN_m,\AN_m)&=\frac{1}{6}-\frac{1}{6m}.
    \end{align*}
    
    Now, assume that there exists a positionwise extension $\hat{d}^*$
    that satisfies $\ID$-, $\AN$, and $\UN$-consistency.  This means
    that both $\hat{d}^*(\UN_{m,n}, \ID_{m,n})$ and
    $\hat{d}^*(\UN_{m,n}, \AN_{m,n})$ have to be constant for
    different values of $m$ and $n$ (as $\hat{d}^*$ must satisfy the
    triangle inequality in addition to the consistency properties).
    However, since we have shown that
    ${d}_\pos(\UN_m,\ID_m)/{d}_\pos(\UN_m,\AN_m)$ is not
    constant for different values of $m$, this is not possible if
    $d^*$ is a positionwise extension.
\end{proof}
}

\appendixproof{Proposition}{prop:swap-triangle}{%
\begin{proof}
  Let $E= (C,V)$ be an election with three candidates, $C=\{a,b,c\}$,
  and six voters:
  \begin{align*}
    v_1 \colon a \pref b \pref c, && v_3 \colon  b \pref c \pref a, && v_5 \colon c \pref a \pref b,\\
    v_2 \colon a \pref b \pref c, && v_4 \colon  b \pref c \pref a, && v_6 \colon c \pref a \pref b,
  \end{align*}
  and let $\ID_{3,6}$ be an identity election with six voters, each with
  preference order $x \pref y \pref z$. One can verify that
  $\hat{d}^\del_\swap(E,\ID_{3,6}) = 8/9$ (indeed, there are no candidates
  to delete, it takes $8$ swaps to make all votes identical in $E$,
  and the normalizing factor is $9$).

  Next, let us consider election $\ID_{2,6}$, where each of the six voters
  has preference order $x \pref y$. For each election $E'$ obtained
  from $E$ by deleting a single candidate, we have
  $\hat{d}^\del_\swap(E',\ID_{2,6}) = 2/3$ (it suffices to make two swaps
  in $E'$ to ensure that all the votes are identical, and the
  normalizing factor is $3$). We also see that
  $\hat{d}^\del_\swap(\ID_{2,6},\ID_{3,6}) = 0$.  Then,
  $\hat{d}^\del_\swap(E,\ID_{3,6}) > \hat{d}^\del_\swap(E,\ID_{2,6}) +
  \hat{d}^\del_\swap(\ID_{2,6},\ID_{3,6})$ and, so, $\hat{d}^\del_\swap$ fails
  the triangle inequality.
\end{proof}
}

\appendixproof{Theorem}{thm:pos-ext}{%

  We first show that the Wasserstein distance between matrices is
  invariant to stretching these matrices.

\begin{lemma}\label{lem:pos-resize}
  Let $X = [\vec{x}_1, \ldots, \vec{x}_m]$ and $Y = [\vec{y}_1, \ldots, \vec{y}_m]$ be
  two frequency matrices, and let $t$ be a nonnegative integer. Then
  $
    d_{W}(X,Y) = d_{W}(\stretched_{mt}(X), \stretched_{mt}(Y)).
  $  
\end{lemma}
\begin{proof}
  To prove this result, we need to recall the way of computing the
  Wasserstein distance between two matrices, provided by
  \citetapp{szu-fal-sko-sli-tal:c:map}. Let
  $A = [\vec{a}_1, \ldots, \vec{a}_k]$ and
  $B = [\vec{b}_1, \ldots, \vec{b}_k]$ be two matrices, whose column
  vectors have nonnegative entries that sum up to $1$ (for each
  vector). To compute $d_{W}(A,B)$, we form flow network $N(A,B)$ that
  contains source $s$, nodes $a_1, \ldots, a_k$, nodes
  $b_1, \ldots, b_k$, and sink $t$ (by a slight abuse of notation, we
  use the same names for both the nodes of the network and the vectors
  in respective matrices, except that for the nodes we omit the arrows
  on top of them). For each $i \in [k]$, there is an arc from $s$ to
  $a_i$ and an arc from $b_i$ to $t$ (each of these arcs has capacity
  $1$ and cost $0$). Further, for each $i, j \in [k]$, there is an arc
  from $a_i$ to $b_j$ with capacity $1$ and cost
  $\frac{1}{k}W(\vec{a}_i,\vec{b}_j)$.  One can verify that in this
  network the min-cost flow of value $k$ from $s$ to $t$ has cost
  equal to $d_{W}(A,B)$.

  Next, consider networks $N_m = N(X,Y)$ and
  $M_{mt} = N(\stretched_{mt}(X), \stretched_{mt}(Y))$.  For each
  $i \in [m]$ and $p \in [t]$, we write $\vec{x}_i^{(p)}$ to refer to the
  $p$-th copy of vector $\vec{x}_i$ that $\stretched_{mt}(X)$ contains. The
  meaning of $\vec{y}_i^{(p)}$ is analogous.
  Let $f_m$ be a min-cost flow of value $m$ from $s$ to $t$ in $N_{m}$
  and, analogously, let $f_{mt}$ be a min-cost flow of value $mt$ from
  $s$ to $t$ in $N_{mt}$. We claim that the costs of $f_m$ and $f_{mt}$
  are equal (and, hence,
  $d_{W}(X,Y) = d_{W}(\stretched_{mt}(X), \stretched_{mt}(Y))$).
  First, we observe that the cost of $f_{mt}$ is greater or equal to
  that of $f_m$. To see why this is the case, consider flow $g$ in
  $N_m$ defined as follows:
  \begin{enumerate}
  \item For each $i \in [m]$ there is flow of value $1$ from $s$ to $x_i$, and
    from $y_i$ to $t$.
  \item For each $i, j \in [m]$, there is flow of value
    $\sum_{p \in [t], q \in [t]} \frac{1}{t}f_{mt}(x_i^{(p)},
    y_j^{(q)})$ from $x_i$ to $y_j$.
  \end{enumerate}
  Note that since $f_{mt}$ is a correct flow, so is $g$, and indeed
  $g$ moves $m$ units of flow from $s$ to $t$. Further, the cost of
  $g$ in $N_m$ is the same as the cost of $f_{mt}$ in $N_{mt}$. Since
  $f_m$ is a min-cost flow of value $m$, its cost cannot be larger
  than that of $g$. An analogous argument shows that the cost of
  $f_{m}$ is greater or equal to that of $f_{mt}$. Hence the two flows
  have equal costs, which proves the lemma.
\end{proof}

Next, we show that $\hat{d}_\pos$ is a pseudodistance.

\begin{proposition}\label{prop:pos-triangle}
  Function $\hat{d}_\pos$ is a pseudodistance. 
\end{proposition}
\begin{proof}
  It is immediate to see that for each two elections $E$ and $F$, we
  have $\hat{d}_\pos(E,E) = 0$, $\hat{d}_\pos(E,F) \geq 0$, and
  $\hat{d}_\pos(E,F) = \hat{d}_\pos(F,E)$. It remains to argue that we
  also have the triangle inequality.

  Let $A$, $B$, and $C$ be three elections with, respectively, $p$,
  $q$, and $r$ candidates.  Further, let $P$, $Q$, and $R$ be their
  frequency matrices. Our goal is to show that:
  \[
    \hat{d}_\pos(A,B) + \hat{d}_\pos(B,C) \geq \hat{d}_\pos(A,C). 
  \]
  By the definition of $\hat{d}_\pos$ and \Cref{lem:pos-resize}, this
  is equivalent to:
  \[
    {d}_{W}(\stretched_{pqr}(P),\stretched_{pqr}(Q)) +
    {d}_{W}(\stretched_{pqr}(Q),\stretched_{pqr}(R)) \geq
    {d}_{W}(\stretched_{pqr}(P),\stretched_{pqr}(R)).
  \]
  Since
  ${d}_{W}$ satisfies the triangle inequality (this has been shown by
  \citetapp{szu-fal-sko-sli-tal:c:map}), this inequality must hold.
\end{proof}
}

\appendixproof{Proposition}{prop:empkem:convergence}{%
\begin{proof}
    We begin by showing a useful fact about the uniformity elections with a large number of candidates.
    It turns out that in such elections every vote is in approximately half of a maximal swap distance to almost every other vote.
    \begin{lemma}
        \label{lemma:empkem:convergence}
        For every $\varepsilon, d > 0$, there exists $M \in \mathbb{N}$ such that for every $m \ge M$, $(C,V) = \UN_{m,m!}$, and $u \in V$, it holds that
        \[
        \left|\left\{v \in V :
        \frac{\swap(u,v)}{\binom{m}{2}} \not \in \left(
            \frac{1}{2}-\varepsilon,
            \frac{1}{2}+\varepsilon 
        \right) \right\}\right| 
        \le d \cdot m!.
        \]
    \end{lemma}
    \begin{proof}
        Let us denote the left hand side of the inequality in the lemma statement by $L_m$.
        Fix $m \in \mathbb{N}$, $(C,V) = \UN_{m,m!}$, and $u \in V$.
        Consider drawing a random vote $v$ from $V$ and random variable $x = \swap(u,v)$.
        This variable follows so called \emph{Mahonian distribution}~\citepapp{ben:x:mahonian}
        and it is known that its mean and variance are equal to 
        \[
            \mu =\frac{m(m-1)}{4} 
            \quad \mbox{and} \quad
            \sigma^2=\frac{m(m-1)(2m+5)}{72},
        \]
        respectively.
        Then, from Chebyshev's inequality for the random variable $x/\binom{m}{2}$, for every $\varepsilon>0$, it holds that
        \begin{align*}
            \mathbb{P}\left(\left|\frac{x}{\binom{m}{2}}-\frac{\mu}{\binom{m}{2}}\right| \ge \varepsilon \right) &\le
            \frac{\sigma^2}{\binom{m}{2}^2\varepsilon^2} = \frac{O(m^{-1})}{\varepsilon^2}.
        \end{align*}
        Since $\mu/\binom{m}{2}=\nicefrac{1}{2}$, we get that
        $|\nicefrac{x}{\binom{m}{2}}-\nicefrac{\mu}{\binom{m}{2}}| \ge \varepsilon$ is equivalent to
        $\nicefrac{\swap(u,v)}{\binom{m}{2}} \not \in (
            \nicefrac{1}{2}-\varepsilon,
            \nicefrac{1}{2}+\varepsilon 
        )$.
        Thus, we obtain
        \(
            \nicefrac{L_m}{m!} \le
            \nicefrac{O(m^{-1})}{\varepsilon^2}.
        \)
        The fact that $O(m^{-1})$ becomes arbitrarily small as $m$ grows yields the thesis.
    \end{proof}
    In the remainder of the proof, let us show that for every $\varepsilon>0$,
    there exists $M \in \mathbb{N}$ such that for every $m \ge M$, it holds that
    \(
        \empkem_i(\UN_{m,m!})/(m!\binom{m}{2})\ge \nicefrac{1}{2} - \varepsilon.
    \)
    Since for every $m$ we have $\empkem_i(\UN_{m,m!})\le \nicefrac{1}{2}\cdot m! \binom{m}{2}$ (this is the maximum Kemeny distance), this will imply the thesis.

    Fix $\varepsilon>0$.
    Let us take arbitrary $m \in \mathbb{N}$ such that
    \(
        \empkem_i(\UN_{m,m!})/(m!\binom{m}{2}) < \nicefrac{1}{2} - \varepsilon.
    \)
    By definition, this means that there exist votes $v_1,\dots,v_i \in V$ such that
    \begin{equation}
    \label{eq:empkem:convergence}
        \frac{1}{m!}\sum_{v \in V} \min_{j \in [i]} \frac{\swap(v,v_j)}{\binom{m}{2}} < \frac{1}{2} - \varepsilon.
    \end{equation}
    Let us denote the left hand side of the above inequality by $L$.
    Next, let us split the set of voters, $V$, into two disjoint sets, $V^+ = \{ v\in V : \min_{j \in [i]}\swap(v,v_j)/\binom{m}{2} \ge \nicefrac{1}{2} - \nicefrac{\varepsilon}{2}\}$ and $V^- = V \setminus V^+$.
    Observe that
    \begin{align*}
        L &= \frac{1}{m!}\left(\sum_{v \in V^+} \min_{j \in [i]} \frac{\swap(v,\! v_j)}{\binom{m}{2}} + \sum_{v \in V^-}\min_{j \in [i]} \frac{\swap(v,\! v_j)}{\binom{m}{2}}\!\right)\\
        &\ge \frac{|V^+|}{m!}\left(\frac{1}{2} - \frac{\varepsilon}{2}\right).
    \end{align*}
    Combining this with inequality \eqref{eq:empkem:convergence} we get
    \[
        |V^+| \le m! \cdot \frac{\frac{1}{2} - \varepsilon}{\frac{1}{2}-\frac{\varepsilon}{2}}.
    \]
    Hence, there exists $\varepsilon' = \varepsilon / (1-\varepsilon) > 0$ such that
    $|V^+| \le m! \cdot (1 - \varepsilon')$.
    Thus, $|V^-| > m! \cdot \varepsilon'$, i.e.,
    \[
        \left|\left\{v \in V : \min_{j \in [i]}
        \frac{\swap(u,v)}{\binom{m}{2}} <
            \frac{1}{2}-\frac{\varepsilon}{2} \right\}\right| 
        > m!\cdot \varepsilon'.
    \]
    Hence, there exists $j \in [i]$ such that
    \[
        \left|\left\{v \in V :
        \frac{\swap(u,v)}{\binom{m}{2}} <
            \frac{1}{2}-\frac{\varepsilon}{2} \right\}\right| 
        > m!\cdot \frac{\varepsilon'}{i}.
    \]
    By Lemma~\ref{lemma:empkem:convergence},
    there exists $M \in \mathbb{N}$ such that for every $m \ge M$ the above inequality does not hold.
    Therefore, for each such $m$, it holds that
    \(
        \empkem_i(\UN_{m,m!})/(m!\binom{m}{2})\ge \nicefrac{1}{2} - \varepsilon,
    \)
    which concludes the proof.
  \end{proof}
}

\section{Missing Proofs}\label{app:proofs}

In this appendix,
we present the technical proofs that were omitted in the main part of the paper.

\appendixProofs

\begin{table*}[t]
    \centering
    \small
    \begin{tabular}{|c|c|c|c|c|c|c|c|c|}
    \toprule
    \begin{sideways} Election Name \end{sideways} &
    \begin{sideways} Marker \end{sideways} &
    \begin{sideways} Preflib File \end{sideways} &
    \begin{sideways} \#Candidates \end{sideways} &
    \begin{sideways} \#Votes \end{sideways} &
    \begin{sideways} Diversity \end{sideways} &
    \begin{sideways} Agreement \end{sideways} &
    \begin{sideways} Polarization \end{sideways} &
    \begin{sideways} Closest Culture\end{sideways}\\
    \midrule
Irish (1) & \tabimg{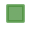} & 00001-00000001.soi & 12 & 43942 & $0.42^*$ & $0.41$ & $0.16^*$ & Urn, $\alpha=0.24$ (0.05) \\ 
Irish (2) & \tabimg{Irish} & 00001-00000002.soi & 9 & 29988 & $0.46^*$ & $0.33$ & $0.18^*$ & Urn, $\alpha=0.3$ (0.01) \\ 
Irish (3) & \tabimg{Irish} & 00001-00000003.soi & 14 & 64081 & $0.37^*$ & $0.48$ & $0.14^*$ & Urn, $\alpha=0.22$ (0.1) \\ 
Debian (1) & \tabimg{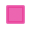} & 00002-00000001.soi & 4 & 475 & $0.28$ & $0.51$ & $0.16$ & Urn, $\alpha=0.29$ (0.03) \\ 
Debian (2) & \tabimg{Debian} & 00002-00000002.soi & 5 & 488 & $0.35$ & $0.47$ & $0.16$ & Urn, $\alpha=0.29$ (0.07) \\ 
Debian (3) & \tabimg{Debian} & 00002-00000003.soi & 7 & 504 & $0.44$ & $0.43$ & $0.1$ & Cube (0.07) \\ 
Debian (4) & \tabimg{Debian} & 00002-00000004.soi & 8 & 421 & $0.4$ & $0.51$ & $0.06$ & Mallows, norm-$\phi=0.46$ (0.07) \\ 
Debian (5) & \tabimg{Debian} & 00002-00000005.soi & 9 & 482 & $0.51$ & $0.35$ & $0.13$ & 5-Cube (0.03) \\ 
Debian (6) & \tabimg{Debian} & 00002-00000006.soi & 5 & 436 & $0.32$ & $0.55$ & $0.1$ & Urn, $\alpha=0.69$ (0.1) \\ 
Debian (7) & \tabimg{Debian} & 00002-00000007.soi & 4 & 403 & $0.2$ & $0.67$ & $0.09$ & Urn, $\alpha=1.39$ (0.06) \\ 
Debian (8) & \tabimg{Debian} & 00002-00000008.soi & 8 & 143 & $0.5$ & $0.35$ & $0.13$ & Cube (0.04) \\ 
Netflix (1) & \tabimg{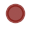} & 00004-00000013.soc & 3 & 617 & $0.2$ & $0.48$ & $0.25$ & Urn, $\alpha=1.88$ (0.05) \\ 
Netflix (2) & \tabimg{Netflix} & 00004-00000092.soc & 3 & 1631 & $0.31$ & $0.26$ & $0.32$ & Square (0.07) \\ 
Netflix (3) & \tabimg{Netflix} & 00004-00000111.soc & 4 & 390 & $0.32$ & $0.46$ & $0.17$ & Urn, $\alpha=0.49$ (0.09) \\ 
Netflix (4) & \tabimg{Netflix} & 00004-00000123.soc & 4 & 376 & $0.37$ & $0.32$ & $0.29$ & Square (0.02) \\ 
Netflix (5) & \tabimg{Netflix} & 00004-00000147.soc & 4 & 400 & $0.45$ & $0.25$ & $0.24$ & Cube (0.03) \\ 
Netflix (6) & \tabimg{Netflix} & 00004-00000148.soc & 4 & 485 & $0.29$ & $0.49$ & $0.18$ & Urn, $\alpha=0.29$ (0.06) \\ 
Netflix (7) & \tabimg{Netflix} & 00004-00000164.soc & 4 & 512 & $0.35$ & $0.32$ & $0.3$ & Square (0.02) \\ 
Netflix (8) & \tabimg{Netflix} & 00004-00000178.soc & 4 & 366 & $0.22$ & $0.61$ & $0.12$ & Urn, $\alpha=0.38$ (0.08) \\ 
Netflix (9) & \tabimg{Netflix} & 00004-00000179.soc & 4 & 454 & $0.25$ & $0.56$ & $0.16$ & Urn, $\alpha=0.29$ (0.07) \\ 
Netflix (10) & \tabimg{Netflix} & 00004-00000186.soc & 4 & 417 & $0.27$ & $0.54$ & $0.14$ & Urn, $\alpha=0.29$ (0.08) \\ 
Skating (1) & \tabimg{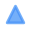} & 00006-00000003.soc & 14 & 9 & $0.04$ & $0.92$ & $0.03$ & Mallows, norm-$\phi=0.06$ (0.04) \\ 
Skating (2) & \tabimg{Skating} & 00006-00000008.soc & 23 & 9 & $0.04$ & $0.94$ & $0.01$ & Mallows, norm-$\phi=0.06$ (0.02) \\ 
Skating (3) & \tabimg{Skating} & 00006-00000011.soc & 20 & 9 & $0.07$ & $0.9$ & $0.02$ & Mallows, norm-$\phi=0.1$ (0.03) \\ 
Skating (4) & \tabimg{Skating} & 00006-00000018.soc & 24 & 9 & $0.03$ & $0.95$ & $0.01$ & Mallows, norm-$\phi=0.04$ (0.01) \\ 
Skating (5) & \tabimg{Skating} & 00006-00000028.soc & 24 & 9 & $0.13$ & $0.85$ & $0.04$ & Mallows, norm-$\phi=0.14$ (0.03) \\ 
Skating (6) & \tabimg{Skating} & 00006-00000029.soc & 19 & 9 & $0.11$ & $0.86$ & $0.04$ & Mallows, norm-$\phi=0.14$ (0.04) \\ 
Skating (7) & \tabimg{Skating} & 00006-00000033.soc & 23 & 9 & $0.06$ & $0.9$ & $0.03$ & Mallows, norm-$\phi=0.1$ (0.05) \\ 
Skating (8) & \tabimg{Skating} & 00006-00000036.soc & 18 & 9 & $0.17$ & $0.76$ & $0.08$ & Single Peaked, Walsh (0.01) \\ 
Skating (9) & \tabimg{Skating} & 00006-00000046.soc & 30 & 7 & $0.04$ & $0.93$ & $0.02$ & Mallows, norm-$\phi=0.06$ (0.02) \\ 
Skating (10) & \tabimg{Skating} & 00006-00000048.soc & 24 & 9 & $0.05$ & $0.93$ & $0.02$ & Mallows, norm-$\phi=0.06$ (0.02) \\ 
ERS (1) & \tabimg{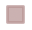} & 00007-00000013.soi & 5 & 2785 & $0.5$ & $0.3$ & $0.16$ & Cube (0.04) \\ 
ERS (2) & \tabimg{ERS} & 00007-00000026.soi & 5 & 148 & $0.41$ & $0.32$ & $0.23$ & Square (0.01) \\ 
ERS (3) & \tabimg{ERS} & 00007-00000029.soi & 17 & 176 & $0.36$ & $0.4$ & $0.23$ & Urn, $\alpha=0.37$ (0.03) \\ 
ERS (4) & \tabimg{ERS} & 00007-00000032.soi & 13 & 575 & $0.51$ & $0.39$ & $0.09$ & 5-Cube (0.04) \\ 
ERS (5) & \tabimg{ERS} & 00007-00000037.soi & 10 & 460 & $0.6$ & $0.28$ & $0.09$ & 10-Cube (0.03) \\ 
ERS (6) & \tabimg{ERS} & 00007-00000050.soi & 10 & 195 & $0.32$ & $0.49$ & $0.16$ & Urn, $\alpha=0.29$ (0.09) \\ 
ERS (7) & \tabimg{ERS} & 00007-00000063.soi & 3 & 213 & $0.32$ & $0.25$ & $0.4$ & Interval (0.13) \\ 
ERS (8) & \tabimg{ERS} & 00007-00000065.soi & 8 & 362 & $0.44$ & $0.29$ & $0.28$ & Square (0.03) \\ 
ERS (9) & \tabimg{ERS} & 00007-00000078.soi & 20 & 365 & $0.55$ & $0.36$ & $0.07$ & 5-Cube (0.05) \\ 
ERS (10) & \tabimg{ERS} & 00007-00000079.soi & 14 & 661 & $0.45$ & $0.47$ & $0.08$ & Mallows, norm-$\phi=0.49$ (0.08) \\ 
Glasgow (1) & \tabimg{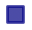} & 00008-00000001.soi & 9 & 6900 & $0.39$ & $0.46$ & $0.17$ & Urn, $\alpha=0.37$ (0.08) \\ 
Glasgow (2) & \tabimg{Glasgow} & 00008-00000005.soi & 10 & 11052 & $0.35$ & $0.5$ & $0.19$ & Urn, $\alpha=0.29$ (0.1) \\ 
Glasgow (3) & \tabimg{Glasgow} & 00008-00000008.soi & 10 & 10160 & $0.32$ & $0.52$ & $0.17$ & Urn, $\alpha=0.29$ (0.08) \\ 
Glasgow (4) & \tabimg{Glasgow} & 00008-00000009.soi & 11 & 9560 & $0.31$ & $0.56$ & $0.16$ & Urn, $\alpha=0.29$ (0.09) \\ 
Glasgow (5) & \tabimg{Glasgow} & 00008-00000010.soi & 9 & 8682 & $0.38$ & $0.45$ & $0.2$ & Urn, $\alpha=0.37$ (0.06) \\ 
Glasgow (6) & \tabimg{Glasgow} & 00008-00000011.soi & 10 & 8984 & $0.39$ & $0.47$ & $0.11$ & Urn, $\alpha=0.22$ (0.12) \\ 
Glasgow (7) & \tabimg{Glasgow} & 00008-00000015.soi & 9 & 8654 & $0.37$ & $0.48$ & $0.17$ & Urn, $\alpha=0.22$ (0.09) \\ 
Glasgow (8) & \tabimg{Glasgow} & 00008-00000018.soi & 9 & 9567 & $0.39$ & $0.49$ & $0.12$ & Mallows, norm-$\phi=0.46$ (0.12) \\ 
Glasgow (9) & \tabimg{Glasgow} & 00008-00000019.soi & 11 & 8803 & $0.32$ & $0.55$ & $0.16$ & Urn, $\alpha=0.29$ (0.1) \\ 
Glasgow (10) & \tabimg{Glasgow} & 00008-00000021.soi & 10 & 5410 & $0.29$ & $0.57$ & $0.14$ & Urn, $\alpha=0.29$ (0.1) \\ 
AGH (1) & \tabimg{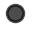} & 00009-00000001.soc & 9 & 146 & $0.36$ & $0.51$ & $0.11$ & Mallows, norm-$\phi=0.41$ (0.12) \\ 
AGH (2) & \tabimg{AGH} & 00009-00000002.soc & 7 & 153 & $0.27$ & $0.59$ & $0.12$ & Urn, $\alpha=0.38$ (0.12) \\ 
Skiing (1) & \tabimg{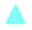} & 00010-00000001.soi & 351 & 4 & $0.18$ & $0.59$ & $0.18$ & Urn, $\alpha=0.38$ (0.04) \\ 
\bottomrule
\end{tabular}
\end{table*}
\begin{table*}
    \centering
    \small
    \begin{tabular}{|c|c|c|c|c|c|c|c|c|}
    \toprule
    \begin{sideways} Election Name \end{sideways} &
    \begin{sideways} Marker \end{sideways} &
    \begin{sideways} Preflib File \end{sideways} &
    \begin{sideways} \#Candidates \end{sideways} &
    \begin{sideways} \#Votes \end{sideways} &
    \begin{sideways} Diversity \end{sideways} &
    \begin{sideways} Agreement \end{sideways} &
    \begin{sideways} Polarization \end{sideways} &
    \begin{sideways} Closest Culture\end{sideways}\\
    \midrule
Skiing (2) & \tabimg{Skiing} & 00010-00000002.soi & 170 & 4 & $0.18$ & $0.58$ & $0.2$ & Urn, $\alpha=0.38$ (0.06) \\ 
Web search (1) & \tabimg{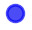} & 00011-00000003.soc & 103 & 5 & $0.15$ & $0.64$ & $0.2$ & Urn, $\alpha=0.38$ (0.05) \\ 
Web search (2) & \tabimg{Web_search} & 00011-00000006.soi & 1449 & 4 & $0.17$ & $0.5$ & $0.26$ & Urn, $\alpha=1.88$ (0.03) \\ 
Web search (3) & \tabimg{Web_search} & 00011-00000008.soi & 1572 & 4 & $0.18$ & $0.45$ & $0.34$ & Urn, $\alpha=0.66$ (0.07) \\ 
Web search (4) & \tabimg{Web_search} & 00011-00000010.soi & 2096 & 4 & $0.2$ & $0.5$ & $0.23$ & Urn, $\alpha=0.29$ (0.05) \\ 
Web search (5) & \tabimg{Web_search} & 00011-00000012.soi & 1210 & 4 & $0.16$ & $0.46$ & $0.36$ & Urn, $\alpha=0.66$ (0.06) \\ 
Web search (6) & \tabimg{Web_search} & 00011-00000038.soi & 1754 & 4 & $0.18$ & $0.46$ & $0.3$ & Urn, $\alpha=1.88$ (0.05) \\ 
Web search (7) & \tabimg{Web_search} & 00011-00000049.soi & 1845 & 4 & $0.18$ & $0.43$ & $0.38$ & Interval (0.08) \\ 
Web search (8) & \tabimg{Web_search} & 00011-00000052.soi & 2242 & 4 & $0.18$ & $0.47$ & $0.29$ & Urn, $\alpha=1.88$ (0.04) \\ 
Web search (9) & \tabimg{Web_search} & 00011-00000054.soi & 2512 & 4 & $0.19$ & $0.49$ & $0.27$ & Urn, $\alpha=1.88$ (0.03) \\ 
Web search (10) & \tabimg{Web_search} & 00011-00000071.soi & 2127 & 4 & $0.18$ & $0.46$ & $0.31$ & Urn, $\alpha=1.88$ (0.06) \\ 
TShirt & \tabimg{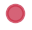} & 00012-00000001.soc & 11 & 30 & $0.46$ & $0.43$ & $0.09$ & Cube (0.09) \\ 
Sushi (1) & \tabimg{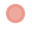} & 00014-00000001.soc & 10 & 5000 & $0.54$ & $0.32$ & $0.14$ & 5-Cube (0.02) \\ 
Sushi (2) & \tabimg{Sushi} & 00014-00000002.soi & 100 & 5000 & $0.26$ & $0.81$ & $0.02$ & Mallows, norm-$\phi=0.23$ (0.05) \\ 
Minneapolis (1) & \tabimg{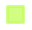} & 00018-00000001.soi & 379 & 36655 & $0.01$ & $0.99$ & $0$ & Mallows, norm-$\phi=0.01$ (0) \\ 
Minneapolis (2) & \tabimg{Minneapolis} & 00018-00000002.soi & 9 & 36655 & $0.33$ & $0.56$ & $0.12$ & Mallows, norm-$\phi=0.39$ (0.12) \\ 
Minneapolis (3) & \tabimg{Minneapolis} & 00018-00000003.soi & 477 & 32086 & $0$ & $0.99$ & $0$ & Urn, $\alpha=3.62$ (0.01) \\ 
Minneapolis (4) & \tabimg{Minneapolis} & 00018-00000004.soi & 7 & 32086 & $0.27$ & $0.6$ & $0.09$ & Urn, $\alpha=0.38$ (0.12) \\ 
MTurk Dots (1) & \tabimg{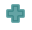} & 00024-00000001.soc & 4 & 795 & $0.49$ & $0.18$ & $0.26$ & Cube (0.06) \\ 
MTurk Dots (2) & \tabimg{MTurk_Dots} & 00024-00000002.soc & 4 & 794 & $0.46$ & $0.25$ & $0.22$ & Cube (0.03) \\ 
MTurk Dots (3) & \tabimg{MTurk_Dots} & 00024-00000003.soc & 4 & 800 & $0.41$ & $0.36$ & $0.15$ & Cube (0.03) \\ 
MTurk Dots (4) & \tabimg{MTurk_Dots} & 00024-00000004.soc & 4 & 794 & $0.37$ & $0.42$ & $0.15$ & Urn, $\alpha=0.22$ (0.05) \\ 
MTurk Puzzle (1) & \tabimg{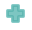} & 00025-00000001.soc & 4 & 793 & $0.48$ & $0.22$ & $0.23$ & Cube (0.03) \\ 
MTurk Puzzle (2) & \tabimg{MTurk_Puzzle} & 00025-00000002.soc & 4 & 795 & $0.37$ & $0.42$ & $0.15$ & Urn, $\alpha=0.22$ (0.05) \\ 
MTurk Puzzle (3) & \tabimg{MTurk_Puzzle} & 00025-00000003.soc & 4 & 795 & $0.4$ & $0.38$ & $0.16$ & Urn, $\alpha=0.24$ (0.03) \\ 
MTurk Puzzle (4) & \tabimg{MTurk_Puzzle} & 00025-00000004.soc & 4 & 797 & $0.44$ & $0.28$ & $0.2$ & Cube (0.04) \\ 
APA (1) & \tabimg{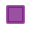} & 00028-00000001.soi & 5 & 18723 & $0.53$ & $0.24$ & $0.16$ & Cube (0.04) \\ 
APA (2) & \tabimg{APA} & 00028-00000002.soi & 5 & 17469 & $0.54$ & $0.16$ & $0.3$ & Urn, $\alpha=0.31$ (0.09) \\ 
APA (3) & \tabimg{APA} & 00028-00000003.soi & 5 & 20239 & $0.51$ & $0.25$ & $0.21$ & Cube (0.04) \\ 
APA (4) & \tabimg{APA} & 00028-00000004.soi & 5 & 17911 & $0.53$ & $0.22$ & $0.21$ & Cube (0.04) \\ 
APA (5) & \tabimg{APA} & 00028-00000006.soi & 5 & 17956 & $0.41$ & $0.41$ & $0.16$ & Urn, $\alpha=0.22$ (0.03) \\ 
APA (6) & \tabimg{APA} & 00028-00000008.soi & 5 & 14506 & $0.5$ & $0.26$ & $0.22$ & Cube (0.05) \\ 
APA (7) & \tabimg{APA} & 00028-00000009.soi & 5 & 16836 & $0.51$ & $0.26$ & $0.18$ & Cube (0.05) \\ 
APA (8) & \tabimg{APA} & 00028-00000010.soi & 5 & 13318 & $0.53$ & $0.22$ & $0.24$ & Cube (0.05) \\ 
APA (9) & \tabimg{APA} & 00028-00000011.soi & 5 & 18286 & $0.51$ & $0.24$ & $0.22$ & Cube (0.05) \\ 
APA (10) & \tabimg{APA} & 00028-00000012.soi & 5 & 15313 & $0.53$ & $0.2$ & $0.28$ & Cube (0.06) \\ 
Labour Party & \tabimg{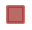} & 00030-00000001.soi & 5 & 266 & $0.36$ & $0.45$ & $0.14$ & Urn, $\alpha=0.22$ (0.1) \\ 
Cujae (1) & \tabimg{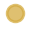} & 00032-00000001.soi & 6 & 32 & $0.39$ & $0.41$ & $0.16$ & Urn, $\alpha=0.22$ (0.05) \\ 
Cujae (2) & \tabimg{Cujae} & 00032-00000002.soc & 6 & 15 & $0.31$ & $0.5$ & $0.17$ & Urn, $\alpha=0.29$ (0.07) \\ 
Cujae (3) & \tabimg{Cujae} & 00032-00000003.soi & 6 & 47 & $0.43$ & $0.36$ & $0.19$ & Cube (0.03) \\ 
Cujae (4) & \tabimg{Cujae} & 00032-00000005.soi & 13 & 14 & $0.28$ & $0.62$ & $0.1$ & Mallows, norm-$\phi=0.33$ (0.1) \\ 
Cities (1) & \tabimg{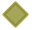} & 00034-00000001.soi & 36 & 392 & $0.43$ & $0.69$ & $0.06$ & Mallows, norm-$\phi=0.36$ (0.09) \\ 
Cities (2) & \tabimg{Cities} & 00034-00000002.soi & 48 & 392 & $0.35$ & $0.76$ & $0.04$ & Mallows, norm-$\phi=0.31$ (0.09) \\ 
Breakfast (1) & \tabimg{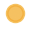} & 00035-00000002.soc & 15 & 42 & $0.53$ & $0.31$ & $0.18$ & 5-Cube (0.05) \\ 
Breakfast (2) & \tabimg{Breakfast} & 00035-00000003.soc & 15 & 42 & $0.47$ & $0.47$ & $0.05$ & Mallows, norm-$\phi=0.51$ (0.05) \\ 
Breakfast (3) & \tabimg{Breakfast} & 00035-00000004.soc & 15 & 42 & $0.57$ & $0.27$ & $0.19$ & 5-Cube (0.03) \\ 
Breakfast (4) & \tabimg{Breakfast} & 00035-00000005.soc & 15 & 42 & $0.51$ & $0.29$ & $0.24$ & Cube (0.05) \\ 
Breakfast (5) & \tabimg{Breakfast} & 00035-00000006.soc & 15 & 42 & $0.52$ & $0.32$ & $0.16$ & 5-Cube (0.03) \\ 
Breakfast (6) & \tabimg{Breakfast} & 00035-00000007.soc & 15 & 42 & $0.51$ & $0.4$ & $0.09$ & 5-Cube (0.05) \\ 
Cycling (1) & \tabimg{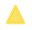} & 00043-00000002.soi & 27 & 12 & $0.39$ & $0.53$ & $0.12$ & Mallows, norm-$\phi=0.41$ (0.11) \\ 
Cycling (2) & \tabimg{Cycling} & 00043-00000010.soi & 24 & 10 & $0.35$ & $0.56$ & $0.11$ & Mallows, norm-$\phi=0.41$ (0.1) \\ 
Cycling (3) & \tabimg{Cycling} & 00043-00000025.soi & 54 & 19 & $0.27$ & $0.69$ & $0.08$ & Mallows, norm-$\phi=0.28$ (0.07) \\ 
Cycling (4) & \tabimg{Cycling} & 00043-00000026.soi & 56 & 19 & $0.3$ & $0.69$ & $0.06$ & Mallows, norm-$\phi=0.31$ (0.04) \\ 
Cycling (5) & \tabimg{Cycling} & 00043-00000033.soi & 74 & 21 & $0.28$ & $0.74$ & $0.08$ & Mallows, norm-$\phi=0.28$ (0.07) \\ 
\bottomrule
\end{tabular}
\end{table*}

\begin{table*}
    \centering
    \small
    \begin{tabular}{|c|c|c|c|c|c|c|c|c|}
    \toprule
    \begin{sideways} Election Name \end{sideways} &
    \begin{sideways} Marker \end{sideways} &
    \begin{sideways} Preflib File \end{sideways} &
    \begin{sideways} \#Candidates \end{sideways} &
    \begin{sideways} \#Votes \end{sideways} &
    \begin{sideways} Diversity \end{sideways} &
    \begin{sideways} Agreement \end{sideways} &
    \begin{sideways} Polarization \end{sideways} &
    \begin{sideways} Closest Culture\end{sideways}\\
    \midrule
Cycling (6) & \tabimg{Cycling} & 00043-00000071.soi & 85 & 21 & $0.22$ & $0.76$ & $0.09$ & Single Peaked, Walsh (0.06) \\ 
Cycling (7) & \tabimg{Cycling} & 00043-00000085.soi & 296 & 19 & $0.31$ & $0.51$ & $0.27$ & Urn, $\alpha=0.29$ (0.08) \\ 
Cycling (8) & \tabimg{Cycling} & 00043-00000115.soi & 178 & 22 & $0.22$ & $0.65$ & $0.12$ & Urn, $\alpha=0.38$ (0.09) \\ 
Cycling (9) & \tabimg{Cycling} & 00043-00000117.soi & 71 & 24 & $0.29$ & $0.64$ & $0.07$ & Mallows, norm-$\phi=0.33$ (0.07) \\ 
Cycling (10) & \tabimg{Cycling} & 00043-00000187.soi & 248 & 20 & $0.37$ & $0.5$ & $0.19$ & Urn, $\alpha=0.37$ (0.12) \\ 
Universities (1) & \tabimg{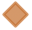} & 00046-00000001.soi & 535 & 18 & $0.23$ & $0.54$ & $0.17$ & Urn, $\alpha=0.29$ (0.06) \\ 
Universities (2) & \tabimg{Universities} & 00046-00000002.soi & 440 & 18 & $0.22$ & $0.6$ & $0.15$ & Urn, $\alpha=0.38$ (0.07) \\ 
Universities (3) & \tabimg{Universities} & 00046-00000003.soi & 1180 & 19 & $0.22$ & $0.57$ & $0.24$ & Urn, $\alpha=0.26$ (0.06) \\ 
Universities (4) & \tabimg{Universities} & 00046-00000004.soi & 1173 & 19 & $0.22$ & $0.58$ & $0.23$ & Urn, $\alpha=0.26$ (0.07) \\ 
Spotify (1) & \tabimg{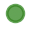} & 00047-00000021.soi & 2242 & 54 & $0.09$ & $0.9$ & $0.02$ & Mallows, norm-$\phi=0.1$ (0.02) \\ 
Spotify (2) & \tabimg{Spotify} & 00047-00000033.soi & 2294 & 54 & $0.09$ & $0.9$ & $0.02$ & Mallows, norm-$\phi=0.1$ (0.02) \\ 
Spotify (3) & \tabimg{Spotify} & 00047-00000098.soi & 2192 & 54 & $0.09$ & $0.89$ & $0.02$ & Mallows, norm-$\phi=0.1$ (0.02) \\ 
Spotify (4) & \tabimg{Spotify} & 00047-00000099.soi & 2146 & 54 & $0.09$ & $0.89$ & $0.02$ & Mallows, norm-$\phi=0.1$ (0.02) \\ 
Spotify (5) & \tabimg{Spotify} & 00047-00000119.soi & 2292 & 54 & $0.09$ & $0.9$ & $0.02$ & Mallows, norm-$\phi=0.1$ (0.02) \\ 
Spotify (6) & \tabimg{Spotify} & 00047-00000146.soi & 2277 & 54 & $0.09$ & $0.9$ & $0.03$ & Mallows, norm-$\phi=0.1$ (0.02) \\ 
Spotify (7) & \tabimg{Spotify} & 00047-00000164.soi & 2413 & 54 & $0.09$ & $0.9$ & $0.02$ & Mallows, norm-$\phi=0.1$ (0.02) \\ 
Spotify (8) & \tabimg{Spotify} & 00047-00000283.soi & 2633 & 53 & $0.08$ & $0.9$ & $0.02$ & Mallows, norm-$\phi=0.1$ (0.02) \\ 
Spotify (9) & \tabimg{Spotify} & 00047-00000307.soi & 2609 & 51 & $0.09$ & $0.9$ & $0.02$ & Mallows, norm-$\phi=0.1$ (0.02) \\ 
Spotify (10) & \tabimg{Spotify} & 00047-00000362.soi & 2681 & 53 & $0.09$ & $0.89$ & $0.02$ & Mallows, norm-$\phi=0.1$ (0.02) \\ 
Movehub & \tabimg{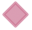} & 00050-00000001.soc & 216 & 12 & $0.47$ & $0.21$ & $0.4$ & Urn, $\alpha=0.31$ (0.11) \\ 
Countries (1) & \tabimg{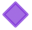} & 00051-00000002.soi & 89 & 17 & $0.57$ & $0.32$ & $0.11$ & 5-Cube (0.03) \\ 
Countries (2) & \tabimg{Countries} & 00051-00000003.soi & 102 & 17 & $0.57$ & $0.31$ & $0.1$ & 5-Cube (0.04) \\ 
Countries (3) & \tabimg{Countries} & 00051-00000004.soi & 110 & 17 & $0.53$ & $0.36$ & $0.11$ & 5-Cube (0.02) \\ 
Countries (4) & \tabimg{Countries} & 00051-00000005.soi & 114 & 19 & $0.61$ & $0.29$ & $0.1$ & 10-Cube (0.02) \\ 
Countries (5) & \tabimg{Countries} & 00051-00000006.soi & 124 & 19 & $0.59$ & $0.32$ & $0.1$ & 10-Cube (0.03) \\ 
Countries (6) & \tabimg{Countries} & 00051-00000007.soi & 146 & 19 & $0.6$ & $0.3$ & $0.12$ & 5-Cube (0.03) \\ 
Countries (7) & \tabimg{Countries} & 00051-00000008.soi & 142 & 19 & $0.58$ & $0.31$ & $0.13$ & 5-Cube (0.02) \\ 
Countries (8) & \tabimg{Countries} & 00051-00000009.soi & 137 & 18 & $0.59$ & $0.29$ & $0.12$ & 5-Cube (0.02) \\ 
Countries (9) & \tabimg{Countries} & 00051-00000010.soi & 145 & 17 & $0.58$ & $0.31$ & $0.13$ & 5-Cube (0.02) \\ 
Countries (10) & \tabimg{Countries} & 00051-00000012.soi & 141 & 15 & $0.55$ & $0.29$ & $0.19$ & 5-Cube (0.05) \\ 
F1 (1) & \tabimg{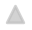} & 00052-00000002.soi & 84 & 8 & $0.2$ & $0.68$ & $0.12$ & Single Peaked, Walsh (0.09) \\ 
F1 (2) & \tabimg{F1} & 00052-00000009.soi & 87 & 11 & $0.19$ & $0.72$ & $0.08$ & Single Peaked, Walsh (0.04) \\ 
F1 (3) & \tabimg{F1} & 00052-00000018.soi & 45 & 11 & $0.25$ & $0.71$ & $0.06$ & Mallows, norm-$\phi=0.28$ (0.05) \\ 
F1 (4) & \tabimg{F1} & 00052-00000023.soi & 42 & 12 & $0.36$ & $0.6$ & $0.09$ & Mallows, norm-$\phi=0.4$ (0.07) \\ 
F1 (5) & \tabimg{F1} & 00052-00000035.soi & 35 & 16 & $0.45$ & $0.51$ & $0.09$ & Mallows, norm-$\phi=0.49$ (0.06) \\ 
F1 (6) & \tabimg{F1} & 00052-00000040.soi & 47 & 16 & $0.35$ & $0.61$ & $0.08$ & Mallows, norm-$\phi=0.36$ (0.07) \\ 
F1 (7) & \tabimg{F1} & 00052-00000042.soi & 41 & 16 & $0.38$ & $0.59$ & $0.06$ & Mallows, norm-$\phi=0.4$ (0.4) \\ 
F1 (8) & \tabimg{F1} & 00052-00000060.soi & 25 & 17 & $0.38$ & $0.56$ & $0.11$ & Mallows, norm-$\phi=0.41$ (0.09) \\ 
F1 (9) & \tabimg{F1} & 00052-00000067.soi & 24 & 21 & $0.37$ & $0.61$ & $0.04$ & Mallows, norm-$\phi=0.39$ (0.02) \\ 
F1 (10) & \tabimg{F1} & 00052-00000068.soi & 25 & 20 & $0.31$ & $0.65$ & $0.06$ & Mallows, norm-$\phi=0.33$ (0.05) \\ 
NSW (1) & \tabimg{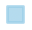} & 00058-00000015.soi & 5 & 47636 & $0.31$ & $0.44$ & $0.21$ & Urn, $\alpha=0.49$ (0.01) \\ 
NSW (2) & \tabimg{NSW} & 00058-00000037.soi & 6 & 49846 & $0.27$ & $0.55$ & $0.17$ & Urn, $\alpha=0.29$ (0.05) \\ 
NSW (3) & \tabimg{NSW} & 00058-00000079.soi & 8 & 48215 & $0.22$ & $0.62$ & $0.13$ & Urn, $\alpha=0.38$ (0.07) \\ 
NSW (4) & \tabimg{NSW} & 00058-00000107.soi & 7 & 45526 & $0.27$ & $0.56$ & $0.15$ & Urn, $\alpha=0.29$ (0.08) \\ 
NSW (5) & \tabimg{NSW} & 00058-00000129.soi & 6 & 48244 & $0.25$ & $0.53$ & $0.24$ & Urn, $\alpha=0.29$ (0.03) \\ 
NSW (6) & \tabimg{NSW} & 00058-00000139.soi & 6 & 47035 & $0.24$ & $0.59$ & $0.17$ & Urn, $\alpha=0.38$ (0.08) \\ 
NSW (7) & \tabimg{NSW} & 00058-00000148.soi & 4 & 51375 & $0.21$ & $0.49$ & $0.33$ & Urn, $\alpha=0.66$ (0.03) \\ 
NSW (8) & \tabimg{NSW} & 00058-00000150.soi & 6 & 50315 & $0.29$ & $0.49$ & $0.2$ & Urn, $\alpha=0.29$ (0.05) \\ 
NSW (9) & \tabimg{NSW} & 00058-00000183.soi & 8 & 47857 & $0.24^*$ & $0.63$ & $0.11^*$ & Urn, $\alpha=0.38$ (0.1) \\ 
NSW (10) & \tabimg{NSW} & 00058-00000235.soi & 7 & 51633 & $0.3$ & $0.54$ & $0.13$ & Urn, $\alpha=0.29$ (0.1) \\ 
\bottomrule
    \end{tabular}
    \caption{The elections from Preflib included in the maps in Figures~\ref{fig:map:preflib} and~\ref{fig:preflib_colors_dap} along with the values of their Diversity, Agreement, and Polarization indices. The values denoted with $*$ has been computed as the average value for 20 elections with 500 votes sampled from the given election.
    In the last column we provide the closest artificial election to a given Preflib election in terms of DAP distance and the distance itself (in the paranthesis).}
    \label{tab:preflib_elections}
\end{table*}

\end{document}